\documentclass[one column,11pt]{IEEEtran}
\pdfoutput=1

\usepackage{amsmath,amssymb,amsfonts,amsthm}
\usepackage{color}
\usepackage{graphicx}
\usepackage{bbm}
\usepackage{algpseudocode} 
\usepackage{algorithm}
\usepackage{subcaption} 

\newtheorem{lemma}{\textbf{Lemma}}

\newtheorem{definition}{\textbf{Definition}}
\newtheorem{theorem}{\textbf{Theorem}}
\newtheorem{proposition}{\textbf{Proposition}}
\newtheorem{example}{\textbf{Example}}

\newtheorem{remark}{\textbf{Remark}}
\newtheorem{case}{\textbf{Case}}

\newcommand{\ie}{{i.e.}}
\newcommand{\eg}{{e.g.}}

\renewcommand{\S}{Section~}
\newcommand{\Apn}{Appendix~}

\newcommand{\mc}[1]{\mathcal{#1}}
\newcommand{\mf}[1]{\mathfrak{#1}}
\newcommand{\mbb}[1]{\mathbb{#1}}
\newcommand{\Bs}[1]{\boldsymbol{#1}}
\newcommand{\Prob}[1]{{\ensuremath{\mathbb{P}\left[ #1 \right]} }}
\newcommand{\Expc}[1]{{\ensuremath{\mathbb{E}\left[ #1 \right]} }}

\newcommand{\1}{\mathbbm{1}}
\newcommand{\Uni}[1]{\mathsf{Uni}\left( #1 \right)}

\newcommand{\set}[1]{\mathcal{#1}}

\newcommand{\ter}{\mathsf{T}}
\newcommand{\alice}{\mathsf{A}}
\newcommand{\bob}{\mathsf{B}}
\newcommand{\calvin}{\mathsf{C}}
\newcommand{\eve}{\mathsf{E}}

\newcommand{\Fbb}{\mathbb{F}}
\newcommand{\Rbb}{\mathbb{R}}

\newcommand{\DoF}{\mathrm{DoF}}

\newcommand{\pktlen}{L}

\DeclareMathOperator{\rank}{rank}

\DeclareMathOperator{\cov}{cov} 

\newcommand{\TableWidth}{0.78\textwidth} 

\title{Group secret key agreement over state-dependent wireless broadcast channels}

\author{Mahdi~Jafari~Siavoshani, Shaunak~Mishra, Christina~Fragouli, Suhas~N.~Diggavi\\ 
\vspace{0.2in} Sharif University of Technology, Tehran, Iran\\
University of California, Los Angeles (UCLA), USA  
\thanks{This work was funded in part by NSF grant 1321120.
    }
}

\begin{document}
\maketitle

\begin{abstract}
We consider a group of $m$ trusted and authenticated nodes that aim to create a shared secret key $K$ over a wireless channel in the presence of an eavesdropper Eve. We assume that there exists a state dependent wireless broadcast channel from one of the honest nodes to the rest of them including Eve. All of the trusted nodes can also discuss over a cost-free, noiseless and unlimited rate public channel which is also overheard by Eve. For this setup, we develop an information-theoretically secure secret key agreement protocol. We show the optimality of this protocol for 
``linear deterministic'' wireless broadcast channels. This model generalizes the packet erasure model studied in literature
for wireless broadcast channels. For ``state-dependent Gaussian'' wireless broadcast channels, we propose an achievability scheme based on a multi-layer wiretap code. Finding the best achievable secret key generation rate leads to solving a non-convex power allocation problem. We show that using a dynamic programming algorithm, one can obtain the best power allocation for this problem. Moreover, we prove the optimality of the proposed achievability scheme for the regime of high-SNR and large-dynamic range over the channel states in the (generalized) degrees of freedom sense.
\end{abstract}

\centerline{\textbf{Keywords}}
Secret key sharing, multi-terminal secrecy, information theoretical secrecy, wireless channel, public discussion.

\section{Introduction}
We consider the problem of generating a secret key $K$ among $m \ge 2$ honest (trusted and authenticated) nodes that communicate over a wireless channel in the presence of a passive eavesdropper Eve. We restrict our attention to the case where communication occurs either through a broadcast channel, where the received symbols are independent among all receivers of the broadcast transmissions including Eve (given that the transmitted symbols is known), or, through a no-cost noiseless public channel.

Here, we focus on the group secret key agreement over a \emph{state-dependent Gaussian broadcast channel}. This model can be motivated by fading wireless channels, where the channel states vary over time; \ie, the variation of SNR\footnote{Signal to noise ratio.} level is modeled by the state of the channel.
The use of state-dependent channels for secrecy has been of interest recently (see for example \cite{khisti_secret-key_2011} and references therein). To gain insight into our problem, we first investigate a \emph{deterministic approximation} of the wireless channel as introduced in \cite{avestimehr_wireless_2011}.
  
For the deterministic broadcast channel we will show that using a superposition 
based secrecy scheme \cite{liang_broadcast_2009}, we can develop a group key agreement protocol that can be shown to be information-theoretically optimal. This can be done by converting the deterministic channel to multiple independent erasure channels. In particular, we show that we 
can get the same key agreement rate for the entire group as we would get for a single pair of nodes. Therefore this result demonstrates that in the presence of an unlimited public channel, we get secret key-agreement rates for linear deterministic channels, that is invariant to network size. 
Similar to the case of erasure broadcast channel \cite{siavoshani_group_2010}, a key idea to get this is a connection to network coding (NC), which allows efficient (in the block length) reconciliation of the group secret (also refer to Appendix~\ref{sec:GrpSecKeyAgr_ErasureBrdcstChnl}
for a review of our previous results on the group secret key agreement over erasure broadcast channels).

We use the deterministic achievability scheme to get some insight about the Gaussian wireless broadcast channel with state. To this end, we use a multi-layer (nested message set, degraded channel) wiretap code based on the broadcast approach of \cite{liang_broadcast_2009,liang_broadcast_2014} to develop a key-agreement protocol for the noisy broadcast problem. This enables a scheme that converts the wireless channel with state to behave similar to the deterministic case. As a result, we show that the achievable secret key generation rate is given by a non-convex optimization problem that determines the power allocation over different layers of the wiretap code. 

Although the power allocation optimization problem is \emph{non-convex}, by investigating and exploiting its special structure, we provide a dynamic programming based algorithm that finds the optimal solution to this optimization problem. The final solution is hard to be written in a closed form expression for the general case. However, the output of our algorithm should not be considered as a numerical approximation but an exact solution. The devised algorithm enables us to evaluate the performance of the proposed group secret key-agreement protocol for various situations.

Finally, we derive an upper bound on the secrecy rate and compare it with the achievable rate by the proposed scheme. Furthermore, we show that although the proposed achievability scheme is not optimal, it can be proved that for the high-SNR regime when there is a large-dynamic range between the channel states, this scheme is optimal in the (generalized) degrees of freedom sense.

\subsection{Related Work}
Secret key generation over wireless channels is a problem that has attracted significant interest.
In a seminal paper on ``wiretap'' channels, Wyner \cite{wyner_wire-tap_1975} pioneered the notion that one can establish information-theoretic secrecy between Alice and Bob by utilizing the noisy broadcast nature of wireless transmissions. 
However, his scheme works only if we have perfect knowledge of Eve's channel 
and moreover, only if Eve has a worse channel than Bob. In a subsequent seminal 
work, Maurer~\cite{maurer_secret_1993} showed the value of feedback from Bob to Alice, even if Eve hears all the feedback transmissions (\ie, the feedback channel is public). He showed that even if the channel from Alice to Eve is better than that to Bob, feedback allows Alice and Bob to create a key which is  information-theoretically secure from Eve. The problem of key agreement between a set of terminals having access to a noisy broadcast channel and a public discussion channel (visible to the eavesdropper) was studied in \cite{csiszar_secrecy_2008}, where the secret key generation capacity is completely characterized, assuming Eve does not have access to the noisy broadcast transmissions. The case when the eavesdropper also had access to the broadcast channel was the main focus of recent work in \cite{gohari_information-theoretic_2010,gohari_information-theoretic_2010-1} which developed upper and lower bounds for secrecy rates.
If the trusted nodes have access to a multi-terminal channel instead of a broadcast channel, \cite{csiszar_secrecy_2013} and \cite{chan_multiterminal_2014}, independently, derived upper and lower bounds for secret key generation capacity under the assumption that Eve has only access to the public channel. 

The best achievable secrecy rate by our scheme for the Gaussian state-dependent channel is given by a non-convex optimization problem (see \eqref{eq:PowerAloocation-P2}) which can be reformulated as a \emph{generalize linear fractional program} \cite{GlobOpt03-GenLinFracProgramming}. In \cite{WirelessCom09-Angela-MAPEL}, the \emph{weighted throughput maximization} problem have been studied which involves a similar optimization problem to \eqref{eq:PowerAloocation-P2} and the authors employs numerical techniques introduced in \cite{GlobOpt03-GenLinFracProgramming} to solve this problem. In our case, however, the convergence time of such numerical method is not practical and we have to develop an approach in \S\ref{sec:GrpSecretKey-GaussBrdcstChnl-SolvePowerAllocProblem} to solve optimization problem \eqref{eq:PowerAloocation-P2} analytically.

To the best of our knowledge, ours is the first work to consider multi-terminal secret key agreement over erasure networks and wireless broadcast channels with state, when Eve also has access to the noisy broadcast transmissions. Moreover, unlike the information-theoretic works (\eg, \cite{wyner_wire-tap_1975,maurer_secret_1993, ahlswede_common_1993, csiszar_secrecy_2008, gohari_information-theoretic_2010-1}) that assume infinite complexity operations, our schemes for the deterministic broadcast channels (that includes the erasure channel case \cite{siavoshani_group_2010}) are computationally efficient. It is worth mentioning that following a conference version of this work  on the packet erasure channel \cite{siavoshani_group_2010}, there has been some attempts to bring those ideas to practical scenarios, \eg, \cite{safaka_exchanging_2011,safaka_exchanging_2013, atsan_low_2013,argyraki_creating_2013}.


The main contributions of the paper can be summarized as follows:
\begin{itemize}
\item For the secret key sharing problem among $m$ trusted nodes that have access to a deterministic broadcast channel and to a public discussion channel, we completely characterize the secret key generation capacity. This result can be considered an extension to the erasure channel case \cite{siavoshani_group_2010} and produces information theoretic secure key regardless of the computational power of the eavesdropper Eve.
\item By using ideas from the code design for deterministic broadcast channels, we devise a coding scheme based on a nested message set, degraded channel wiretap code (also see \cite{liang_broadcast_2009}). In general, we characterize the achievable secret key generation rate which is given by a non-convex power allocation optimization problem. Moreover, we derive an upper bound on the secrecy rate and show that for the high-SNR, high-dynamic range regime, our proposed scheme is optimal in a degree of freedom sense.
\item Although, in the proposed scheme, the best achievable secrecy rate is described by the solution of a non-convex optimization problem, we \emph{solve} the optimization problem by using dynamic programming technique. 
\end{itemize}

The rest of the paper is organized as follows. In \S\ref{sec:Notation_ProblemStatement}, we introduce our notation and the problem formulation. \S\ref{sec:MainResults} summarizes the main results of the paper.
Our general upper bound on the secret key generation capacity for an independent broadcast channel is presented in \S\ref{sec:IndpBrdCstChnl-Secrecy-UpperBound}.
Each of the ``deterministic,'' and ``state-dependent Gaussian'' models will be discussed in  \S\ref{sec:GrpSecKeyAgr_DetBrdcstChnl} and \S\ref{sec:GrpSecKeyAgr_GaussBrdcstChnl}, respectively.
The solution of the non-convex optimization problem is derived in \S\ref{sec:GrpSecretKey-GaussBrdcstChnl-SolvePowerAllocProblem}.
Finally, open questions and future directions will be discussed in \S\ref{sec:OpenQuestion_FutureDirections}.

\section{Notation and Problem Statement}\label{sec:Notation_ProblemStatement}
For convenience, during the paper, we use $[i:j]$ to denote the set of integers $\{i,i+1,\ldots,j\}$. 
Given random variables $X_1,\ldots,X_m$, we write $X_{1:m}$ to denote $\left(X_1,\ldots,X_m \right)$. We use also $X^{t}$ to denote $\left(X[1],\ldots,X[t]\right)$ where $t$ is the discreet time index. 
%
%
%
With an abuse of the notation, we use $H(\cdot)$ to denote both entropy and differential entropy depending on the context.
All the logarithms are in base two unless otherwise stated.
We write $f(x)\stackrel{\cdot}{=}g(x)$ to denote that $\log{f(x)} = \log{g(x)} + o(\log{x})$.
The notation ``$\stackrel{\cdot}{\le}$'' and ``$\stackrel{\cdot}{\ge}$'' are defined similarly.

\subsection{Problem Statement}
We consider a set of $m \ge 2$ honest nodes $\{0, \ldots, m-1\}$ 
that
aim to share a secret key $K$ among themselves while keeping it concealed from a passive adversary Eve, denoted by ``$\eve$''. Eve does not perform any transmissions, but is trying to eavesdrop on (overhear) the communications between the honest nodes\footnote{For convenience, sometimes we will refer to legitimate terminals $0,1,2,\ldots,$ as ``Alice,'' ``Bob,'' ``Calvin,'' and so on. So for example, we use $X_0,X_1,X_2$, etc. interchangeably with $X_\alice,X_\bob,X_\calvin$, etc.}.

We assume that Alice (terminal $0$
) has access to a broadcast channel such that the rest of the terminals (including Eve) receive \textit{independent} noisy version of what she broadcasts (see Figure~\ref{fig:SetupBroadcastChannel}), 
where the input and output symbols of the channel are from some arbitrary sets.
We also assume that all of the honest terminals can discuss over a cost-free noiseless public channel where everybody (including Eve) can hear the discussion (see Figure~\ref{fig:SetupPublicChannel}).

\begin{figure}
\centering
\begin{subfigure}[b]{\columnwidth}
\centering
\includegraphics[scale=0.85]{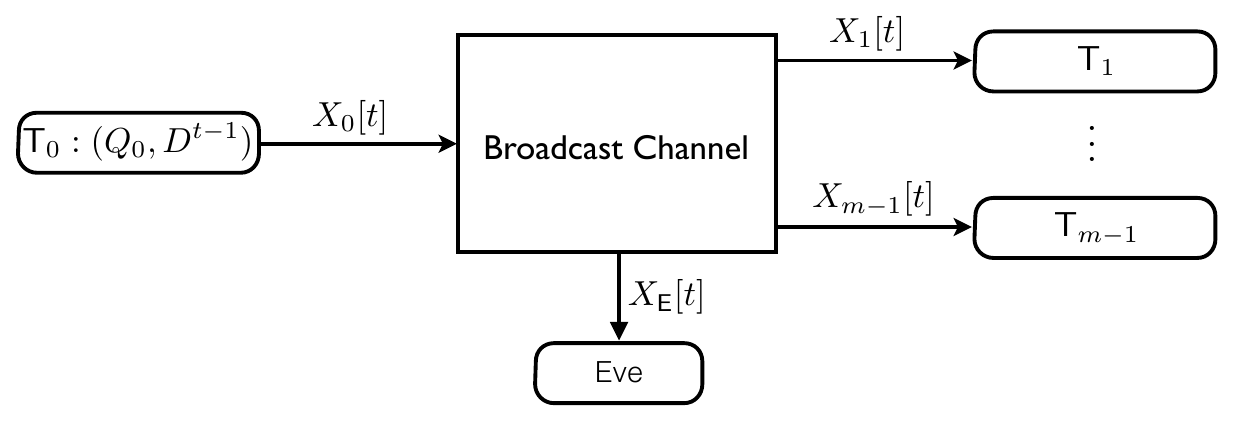}
\caption{}
\label{fig:SetupBroadcastChannel}
\end{subfigure}
\\
\begin{subfigure}[b]{\columnwidth}
\centering
\includegraphics[scale=0.85]{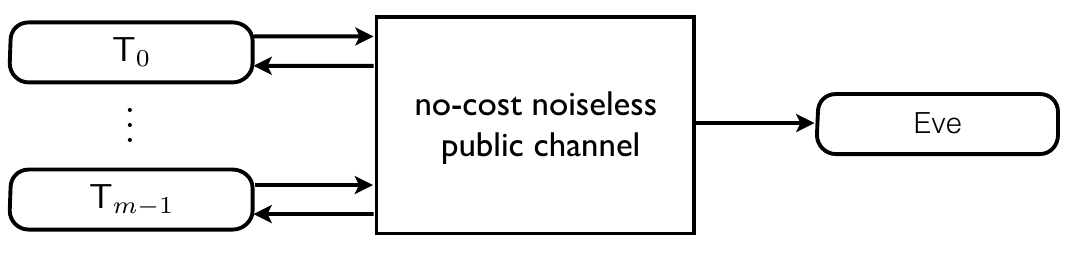}
\caption{}
\label{fig:SetupPublicChannel}
\end{subfigure}
\caption{(a) The broadcast channel that Alice (terminal $\ter_0$) has access to. 
(b) The cost-free noiseless public channel that all of the trusted node can discuss over. Eve overhears all of the public discussions completely.}\label{fig:SetupChannel}
\end{figure}

The protocol stated in Definition~\ref{def:GrpSecKey-SecKeyGenProtocol} 
introduces the most general form of an interactive communication between terminals aiming to share a common secret key $K$ (see also \cite{maurer_secret_1993, ahlswede_common_1993, csiszar_secrecy_2008, gohari_information-theoretic_2010-1}). 

\begin{definition}[Secret key generating protocol]\label{def:GrpSecKey-SecKeyGenProtocol}
~\\
\textrm{1)} For $t=0$, all of the honest terminals generate independent random 
variables $Q_0,\ldots,Q_{m-1}$.\\
\textrm{2)} \textsf{(i)} For time $1\le t\le n$, Alice transmits $X_0[t]$ over the broadcast channel. Then the other terminals receive $X_1[t],\ldots,X_{m-1}[t],$ and Eve receives $X_\eve[t]$.\\
\textsf{(ii)} Following each of the broadcast transmissions, there is the possibility for the legitimate terminals to discuss over a cost-free noiseless public channel. The discussion continues in a round robin order for an arbitrary number of rounds. The whole public discussion at time $t$ is denoted by $D[t]$. Notice that what Alice broadcasts at time $t$ depends on $Q_0$ and $D^{t-1}$.\\
\textrm{3)} Finally, the $i$th terminal creates a key $K_i$ where
$K_i = K_i(Q_i,X_i^n,D^n)$.
\end{definition}

\begin{definition}\label{def:GrpSecKey-AchvKeyRate-Capacity}
A number $R_s$ is called an achievable key generation rate if for every $\epsilon>0$ and sufficiently large $n$ there exists a key generating protocol as defined in Definition~\ref{def:GrpSecKey-SecKeyGenProtocol} such that we have
\begin{gather}
\Prob{K_i\neq K_j}<\epsilon,\quad \forall i,j:i\neq j, \label{eq:AchvRate_Cond1} \\
I(K_0;X_\eve^n,D^n)<\epsilon, \label{eq:AchvRate_Cond2}\\
\frac{1}{n} H(K_0)> R_s-\epsilon.\label{eq:AchvRate_Cond3}
\end{gather}
The supremum of the achievable key rate as $n\rightarrow\infty$ and $\epsilon\rightarrow 0$ is called the \emph{secret key generation (SKG) capacity} $C_s$.
\end{definition}

\subsection{State-Dependent Gaussian Broadcast Channels}
\label{sec:GrpSecKey-GaussBrdcstChnl-ProblemStatement}
Here, we introduce the state-dependent additive white Gaussian broadcast channel model which is the main focus of this paper. In this model, we assume that for each receiver the channel state remains unchanged during  a block of symbols of length $\pktlen$ and changes independently from one block to another block. We also assume $\pktlen$ is large enough so that enables us to apply information theoretical arguments within each block. The transmitted vector sent by Alice is denoted by ${X}_\alice\in\Rbb^{\pktlen}$. The received vector at each receiver (including Eve) depend on its channel state at a particular time instant. We define a random variable $S_i[t] \in [0:s]$ corresponding to the channel state for the $i$th terminal at time $t$ and similarly define the random variable $S_\eve[t] \in [0:s]$ for Eve.  For the channel state of a receiver $r\in\{1,\ldots,m-1,\eve\}$ we assume that\footnote{For simplicity of demonstration and without loss of generality, here we only consider a symmetric problem where the probability distribution over the states are the same for all of the receivers (including Eve). Moreover, we focus on a finite number of states. Both of these restrictions can be relaxed.}
$\Prob{S_r[t] = k} = \delta_k, \forall k\in [0:s],$
where $\sum_{k=0}^s \delta_k=1$.
The received vector at the receiver $r$ is modelled by a state-dependent white Gaussian channel as follows
\begin{equation}\label{eq:GaussianChannelModel}
\hat{X}_r[t] = \sqrt{h_{S_r[t]}} {X}_\alice[t] + {Z}_r[t],\quad \forall r\in\{1,\ldots,m-1,\eve\},
\end{equation}
where $\hat{X}_r[t]\in\Rbb^\pktlen$ and ${Z}_r[t]\in\Rbb^\pktlen$ . 
For the additive noise of each receiver we assume ${Z}_r[t] \sim N(0,\Bs{I}_\pktlen)$ and the noise vectors are also independent over time. The channel gains $\sqrt{h_i}$ are some real constants such that
$h_0 < \cdots < h_s$.
Additionally, the channel input is subject to an average power constraint $P_{\mathrm{max}}$, \ie,
$\frac{1}{\pktlen}\Expc{ \|{X}_\alice \|^2} \le P_{\mathrm{max}}$.

Moreover, we assume that the CSI\footnote{Channel state information.} is completely known by each receiver. So we define a composite received vector for each receiver $r$ as 
${X}_r[t] = (\hat{X}_r[t],S_r[t])$.

\subsection{Deterministic Broadcast Channel}\label{sec:GrpSecKey-DetBrdcstChnl-ProblemStatement}
Now, following the idea proposed in \cite{avestimehr_wireless_2011}, we introduce the deterministic approximation model for our Gaussian channel. 
We assume that the transmitted vector (packet) sent by Alice is denoted by ${X}_\alice\in\Fbb_q^\pktlen$ where $\Fbb_q$ is a finite field of size $q$. 
Then, the received vector at the receiver $r$ is modelled by a state-dependent deterministic broadcast channel as follows
\begin{equation}\label{eq:DeterministicChannelModel}
\hat{{X}}_r[t] = \Bs{F}_{S_r[t]} {X}_\alice[t],\quad\quad \forall r\in\{1,\ldots,m-1,\eve\},
\end{equation}
where $\Bs{F}_i\in\Fbb_q^{\pktlen\times \pktlen}$ for $i\in [0:s]$ and $S_r[t]$ is defined in \S\ref{sec:GrpSecKey-GaussBrdcstChnl-ProblemStatement}. Moreover, similar to the Gaussian model, we define a composite received vector for the receiver $r$ as
${X}_r[t] = ( \hat{{X}}_r[t], S_r[t] )$.

In order to capture and model the different SNR level for the Gaussian channel, we use the shift matrix model developed in \cite{avestimehr_wireless_2011}. To this end, we consider matrices $\Bs{F}_i$ such that they satisfy the following nested structure
\begin{align}
\vec{0}=\ker\Bs{F}_s\subset \ker\Bs{F}_{s-1} \subset \cdots\subset \ker\Bs{F}_0=\Fbb_q^\pktlen, \label{eq:TransMatrixNested_Structure_1} \\
\rank(\Bs{F}_i-\Bs{F}_{i-1}) = \rank(\Bs{F}_i) - \rank(\Bs{F}_{i-1}). \label{eq:TransMatrixNested_Structure_2}
\end{align}
For convenience we assume that $\Bs{F}_s=\Bs{I}_\pktlen$ where $\Bs{I}_\pktlen$ is the identity matrix of size $\pktlen$. The two extreme states ``$0$'' and ``$s$'' correspond to complete erasure and complete reception of the transmitted vector (packet) ${X}_\alice$. The deterministic model is indeed an extension to the packet erasure broadcast channel, studied in \cite{siavoshani_group_2010,argyraki_creating_2013}, which has only two channel states, \ie, $s=1$, (see also  Appendix~\ref{sec:GrpSecKeyAgr_ErasureBrdcstChnl}).

\section{Main Results}\label{sec:MainResults}
The main results of this paper is summarized in the following. For the secret key generation scenario among $m$ terminals that have access to a ``deterministic broadcast channel,'' we completely characterize the key generation capacity. This result can be considered as the generalization of the result of \cite{siavoshani_group_2010,argyraki_creating_2013} for ``packet erasure broadcast channels'' (see Theorem~\ref{thm:GrpScrtKey-DetBrdcstChnl-MainRslt}). For a ``state-dependent Gaussian broadcast channel,'' we provide upper and lower bounds for the key generation capacity and show that these bounds will match in the high-dynamic range, high-SNR regime. 
Furthermore, the achievable secrecy rate by our proposed scheme for the Gaussian model is described by a non-convex power optimization problem. Although this problem is non-convex, by exploiting its special structure, we find the optimal power allocation that leads to the best secrecy rate achievable by the proposed scheme.



\begin{theorem}\label{thm:GrpScrtKey-DetBrdcstChnl-MainRslt}
The SKG capacity among $m$ terminals that have access to a state-dependent deterministic broadcast channel, defined in \S\ref{sec:GrpSecKey-DetBrdcstChnl-ProblemStatement}, is given by
\begin{equation*}
C^{\mathsf{det}}_s = \sum_{i=1}^s \left[ \rank{\Bs{F}_i} - \rank{\Bs{F}_{i-1} }  \right] \theta_i(1-\theta_i) \log{q},
\end{equation*}
where $\theta_i \triangleq \sum_{j=0}^{i-1} \delta_j$.
\end{theorem}

This theorem is proved in \S\ref{sec:GrpSecKeyAgr_DetBrdcstChnl}.
Notice that the result of \cite{siavoshani_group_2010} is a special case of Theorem~\ref{thm:GrpScrtKey-DetBrdcstChnl-MainRslt} when $s=1$.

\begin{theorem}\label{thm:GrpSecKey-GaussBrdcstChnl-MainRslt}
The SKG capacity among $m$ terminals that have access to a state-dependent Gaussian broadcast channel, as defined in \S\ref{sec:GrpSecKey-GaussBrdcstChnl-ProblemStatement}, is upper bounded by
\begin{equation*}
C^{\mathsf{gaus}}_s \le \frac{1}{2} L \sum_{i=0}^{s}\sum_{j=0}^{s} \delta_i\delta_j  \log \left(1 + \frac{h_i P_{\mathrm{max}}}{1 + h_j P_{\mathrm{max}}} \right).
\end{equation*}
Moreover, the secrecy capacity can be lower bounded by the solution of the following (non-convex) optimization problem
\begin{equation*}
C^{\mathsf{gaus}}_s \ge \left\{\begin{array}{ll}
\max & \sum_{i=1}^s \Delta_i L R_i \\
\mathrm{subject\ to} & \sum_{i=1}^s P_i = P_{\mathrm{max}}\\
 & P_i\ge 0,\quad \forall i\in[1:s],
\end{array} \right.
\end{equation*}
where $\Delta_i\triangleq (1-\theta_i)\theta_i$. Also $\forall i\in[1:s]$ we have
\begin{equation*}
R_i \triangleq \frac{1}{2} \left[ \log\left(1+ \frac{h_i P_i}{1+ h_i I_i} \right) -\log\left(1+ \frac{h_{i-1} P_i}{1+ h_{i-1} I_i} \right) \right],
\end{equation*}
where $I_i \triangleq \sum_{j=i+1}^s P_j$. 
%
Additionally, for the high-dynamic range case where $h_i \gg h_{i-1}, \forall i\in[1:s]$, and when we are in high SNR regime, we can write
\begin{equation*}
C^{\mathsf{gaus}}_s \stackrel{\cdot}{=} \frac{1}{2} L \sum_{i=1}^s \Delta_i \log\frac{h_i}{h_{i-1}},
\end{equation*}
where ``$\stackrel{\cdot}{=}$'' defined in \S\ref{sec:Notation_ProblemStatement}, is used to denote for the exponential equality with respect to some scaling parameter $Q$. Here, as $Q\rightarrow\infty$, we asymptotically approach to the high-dynamic, high-SNR regime (for more details refer to \S\ref{sec:HighDynamic_HighSNR_Regime}).
\end{theorem}

It is worth mentioning that the power optimization problem stated in Theorem~\ref{thm:GrpSecKey-GaussBrdcstChnl-MainRslt} is a non-convex problem. 
Although the closed-form solution of the this problem is not easy to derive explicitly, but by using dynamic programming it can be easily found numerically. In \S\ref{sec:GrpSecretKey-GaussBrdcstChnl-SolvePowerAllocProblem}, based on the structure of this optimization problem and by exploiting special properties of its KKT necessary conditions for the optimality, we propose a dynamic programming algorithm that finds the optimal power allocation (see Algorithms~\ref{alg:Solve_KKT_Conditions} and \ref{alg:Solve_KKT_RecursionMainPart}). More specifically, we have the following theorem.

\begin{theorem}
Algorithms~\ref{alg:Solve_KKT_Conditions}~and~\ref{alg:Solve_KKT_RecursionMainPart} find the optimal solution of the optimization problem stated in Theorem~\ref{thm:GrpSecKey-GaussBrdcstChnl-MainRslt}.
\end{theorem}

\section{Upper Bound for the Key Generation Capacity of Independent Broadcast Channels}\label{sec:IndpBrdCstChnl-Secrecy-UpperBound}
The secret key generation capacity among multiple terminals (without 
eavesdropper having access to the broadcast channel) is completely 
characterized in \cite{csiszar_secrecy_2008}. 
By using this result, it is possible to state an upper bound  for the 
secrecy capacity of the key generation problem among multiple terminals 
where the eavesdropper has also access to the broadcast channel.
This can be done by adding a dummy terminal to the first problem and
giving all the eavesdropper's information to this dummy node and let
it to participate in the key generation protocol. By doing so, the 
secret key generation rate does not decrease.
Hence by combining \cite[Theorem~4.1]{csiszar_secrecy_2008} and
\cite[Lemma~5.1]{csiszar_secrecy_2008}, the following result can be 
stated.

\begin{theorem} \label{thm:SecrecyUpBound-CsNa08}
The secret key generation capacity among $m$ terminals as defined in 
Definition~\ref{def:GrpSecKey-AchvKeyRate-Capacity}, is upper bounded as 
follows
\begin{equation*}
C_s\le \max_{P_{X_0}} \min_{\lambda\in\Lambda([0:m-1])} \Bigg[ H(X_{[0:m-1]}|X_\eve)\\ -\sum_{B\subsetneq [0:m-1]} \lambda_B H(X_B|X_{B^c},X_\eve) \Bigg],
\end{equation*}
where $\Lambda([0:m-1])$ is the set of all collections $\lambda=\left\{ \lambda_B : B\subsetneq [0:m-1], B\neq \varnothing \right\}$
of weights $0\le\lambda_B\le 1$, satisfying
\begin{equation}\label{eq:LambdaPartitionDef}
\sum_{B\subsetneq [0:m-1],\ i\in B} \lambda_B = 1, \quad\quad \forall i\in [0:m-1].
\end{equation}
In the above expression for the upper bound, it is possible to 
change the order of maximization and minimization, 
see \cite[Theorem~4.1]{csiszar_secrecy_2008}.
\end{theorem}

\begin{remark}
The upper bound stated in Theorem~\ref{thm:SecrecyUpBound-CsNa08} is not the best known upper bound for the secret key sharing capacity of the multi-terminal secret key sharing problem (see also \cite{gohari_information-theoretic_2010, gohari_information-theoretic_2010-1} for alternative improved bounds). However in this work, we use Theorem~\ref{thm:SecrecyUpBound-CsNa08} to derive an upper bound for our problem. This bound is good enough that in addition to the proposed achievability scheme, completely characterize the secret key sharing capacity for the 
``state-dependent deterministic channels'' scenario.
\end{remark}

Now, back to our problem where the channel from Alice to the other terminals are assumed to be independent, we can further simplify the upper bound given in Theorem~\ref{thm:SecrecyUpBound-CsNa08}, as 
stated in Theorem~\ref{thm:SecrecyUpBound-CsNa08-forIndpChnl}.

\begin{theorem}[See also \cite{siavoshani_group_2010}]\label{thm:SecrecyUpBound-CsNa08-forIndpChnl}
If the channels from Alice to the other terminals are independent, 
then the upper bound stated in Theorem~\ref{thm:SecrecyUpBound-CsNa08} for the 
SKG capacity is simplified to
\begin{align}
C_s &\le \max_{P_{X_0}} \min_{j\in[1:m-1]} I(X_0;X_j|X_\eve) \label{eq:SecrecyUpperBnd-CsNa08-Final-1} \\
 &\le \min_{j\in[1:m-1]} \max_{P_{X_0}} I(X_0;X_j|X_\eve). \label{eq:SecrecyUpperBnd-CsNa08-Final-2} 
\end{align}
\end{theorem}
\begin{proof}
For the proof refer to Appendix~\ref{apn:SomeProofs}.
\end{proof}


\begin{remark}
Using \cite[Theorem~7]{maurer_secret_1993} or \cite[Theorem~2]{ahlswede_common_1993}, we observe 
that the bound given in \eqref{eq:SecrecyUpperBnd-CsNa08-Final-2} is indeed tight for the two terminals problem where we have the Markov chains $X_\bob\leftrightarrow X_\alice \leftrightarrow X_\eve$, \ie,
when the channels are independent or $X_\alice \leftrightarrow X_\bob \leftrightarrow X_\eve$, \ie, when the channels are degraded. 
In \S\ref{sec:GrpSecKeyAgr_DetBrdcstChnl}, we will further show that the above upper bound is also tight for the 
stated-dependent deterministic broadcast channels.
\end{remark}

\section{Group Secret Key Agreement over Deterministic Broadcast Channels}\label{sec:GrpSecKeyAgr_DetBrdcstChnl}
In this section, we prove Theorem~\ref{thm:GrpScrtKey-DetBrdcstChnl-MainRslt} that characterizes the secret key generation capacity for a deterministic broadcast channel defined in \S\ref{sec:GrpSecKey-DetBrdcstChnl-ProblemStatement}.
The proof of this theorem, as an underlying machinery, uses the achievability technique for the packet erasure broadcast channel that is appeared in \cite{siavoshani_group_2010,argyraki_creating_2013} 
(for more details, also see  Appendix~\ref{sec:GrpSecKeyAgr_ErasureBrdcstChnl}). 

\subsection{Upper Bound for the Key Generation Capacity}
Using Theorem~\ref{thm:SecrecyUpBound-CsNa08-forIndpChnl}, the SKG capacity $C^{\mathsf{det}}_s$ for the independent broadcast channel can be upper bounded by \eqref{eq:SecrecyUpperBnd-CsNa08-Final-2}.
%
Then we have the following result, Theorem~\ref{thm:SecKeyGenCapacity_UpperBound_Deterministic}.
\begin{theorem}\label{thm:SecKeyGenCapacity_UpperBound_Deterministic}
The SKG capacity of the deterministic broadcast channel, introduced in \S\ref{sec:GrpSecKey-DetBrdcstChnl-ProblemStatement}, is upper bounded by 
\begin{equation*}
C^{\mathsf{det}}_s \le \sum_{i=1}^s \left[ \rank{\Bs{F}_i} - \rank{\Bs{F}_{i-1} }  \right] \theta_i(1-\theta_i) \log{q},
\end{equation*}
where $\theta_i=\sum_{j=0}^{i-1} \delta_j$.
\end{theorem}

\begin{proof}
From \eqref{eq:SecrecyUpperBnd-CsNa08-Final-2} and because of the symmetry of the problem, we have
$C^{\mathsf{det}}_s \le \max_{P_{{X}_\alice}} I({X}_\alice; {X}_\bob| {X}_\eve)$
where we use ``$\alice$'' and ``$\bob$'' to denote for terminal~$0$ and terminal~$1$. Then, we can write
\begin{align*}
H({X}_\alice| {X}_\eve) &= \sum_{i=0}^{s-1} \delta_i \left[ H({X}_\alice|\hat{X}_\eve,S_\eve=i) \right] \nonumber\\
&= \sum_{i=0}^{s-1} \delta_i \left[ H({X}_\alice,\hat{X}_\eve | S_\eve=i) - H(\hat{X}_\eve|S_\eve=i) \right] \nonumber\\
&= \sum_{i=0}^{s-1} \delta_i \left[ H({X}_\alice)- H(\Bs{F}_i {X}_\alice) \right],
\end{align*}
and similarly
\begin{align*}
H({X}_\alice| {X}_\bob, {X}_\eve)= 
\sum_{i=0}^{s-1} \kappa_i \left[ H({X}_\alice)- H(\Bs{F}_i {X}_\alice) \right],
\end{align*}
where $\kappa_i\triangleq 2\delta_i(\delta_0+\cdots+\delta_{i-1})\1_{\{i>0\}} +\delta_i^2$.
Thus, we have
$I({X}_\alice;{X}_\bob| {X}_\eve) = \sum_{i=0}^{s-1} \rho_i  \left[H({X}_\alice)- H(\Bs{F}_i {X}_\alice) \right]$
where $\rho_i \triangleq \delta_i- \kappa_i$.
Now, by observing that 
$H(\Bs{F}_i {X}_\alice)=H(\Bs{F}_i {X}_\alice, \Bs{F}_{i-1} {X}_\alice)$
and applying the chain rule recursively, we get
$H(\Bs{F}_i {X}_\alice) = \sum_{j=1}^i H(\Bs{F}_j {X}_\alice | \Bs{F}_{j-1} {X}_\alice)$.
So $I({X}_\alice; {X}_\bob| {X}_\eve)$ can be expanded as follows
\begin{align*}
I({X}_\alice; {X}_\bob| {X}_\eve) &= \sum_{i=0}^{s-1} \rho_i [H({X}_\alice) - H(\Bs{F}_i {X}_\alice)] \nonumber\\
  &= \sum_{j=1}^s H(\Bs{F}_j {X}_\alice | \Bs{F}_{j-1} {X}_\alice) \sum_{i=0}^{j-1} \rho_i.
\end{align*}
Hence we can upper bound $C^{\mathsf{det}}_s$ as follows 
{\allowdisplaybreaks[4]
\begin{align}\label{eq:UpperBound_Final}
C^{\mathsf{det}}_s 
&\le \max_{P_{{X}_\alice}} \sum_{j=1}^s H(\Bs{F}_j {X}_\alice| \Bs{F}_{j-1} {X}_\alice) \sum_{i=0}^{j-1} \rho_i \nonumber\\
&= \max_{P_{{X}_\alice}} \sum_{j=1}^s H\left([\Bs{F}_j - \Bs{F}_{j-1} ] {X}_\alice| \Bs{F}_{j-1} {X}_\alice \right) \sum_{i=0}^{j-1} \rho_i \nonumber\\
&\stackrel{\text{(a)}}{\le} \max_{P_{{X}_\alice}} \sum_{j=1}^s H\left([\Bs{F}_j - \Bs{F}_{j-1} ] {X}_\alice \right) \sum_{i=0}^{j-1} \rho_i \nonumber\\
&\stackrel{\text{(b)}}{=} \sum_{j=1}^s \rank\left(\Bs{F}_j - \Bs{F}_{j-1} \right)  \left( \sum_{i=0}^{j-1} \rho_i \right) \log{q} \nonumber\\
&\stackrel{\text{(c)}}{=} \sum_{j=1}^s \left[ \rank{\Bs{F}_j} - \rank{\Bs{F}_{j-1} }  \right] \left( \sum_{i=0}^{j-1} \rho_i \right) \log{q}, 
\end{align}}\hspace{-5pt}
where (a) is true because conditioning reduces the entropy,
(b) is true because uniform distribution on ${X}_\alice$ achieves the maximum
values for all the entropies in the summation, and finally (c) is
true because of the assumption we have made in \eqref{eq:TransMatrixNested_Structure_2}.
Also, note that
$\sum_{i=0}^{j-1} \rho_i = \theta_j(1-\theta_j)\ge 0,$
where $\theta_j \triangleq \sum_{i=0}^{j-1} \delta_i$. This completes the proof.
\end{proof}

\subsection{Lower Bound for the Key Generation Capacity}
In this section, we will present a scheme that achieves the same secret key generation rate as we derived in the upper bound stated in Theorem~\ref{thm:SecKeyGenCapacity_UpperBound_Deterministic}.
But before that, let us state the following proposition.
\begin{proposition}\label{prop:Kernel_DirectSum_DetAchvScheme}
Suppose $s$ subspaces $\ker \Bs{F}_i\subseteq\Fbb_q^L$ satisfy the nested condition \eqref{eq:TransMatrixNested_Structure_1}, \ie, $\vec{0}=\ker\Bs{F}_s\subset \ker\Bs{F}_{s-1} \subset \cdots\subset \ker\Bs{F}_0=\Fbb_q^\pktlen$. Then it is possible to find subspaces $\Pi_1,\ldots,\Pi_s,$ such that $\cap_{i\in\mc{V}} \Pi_i=\vec{0}$ for all $\set{V} \subseteq[1:s]$ where $|\set{V}|\ge 2$ and they also satisfy
\begin{gather}
\Pi_1 \oplus \ker\Bs{F}_1 = \Fbb_q^\pktlen, \nonumber\\
\Pi_2\oplus\Pi_1 \oplus \ker\Bs{F}_2 = \Fbb_q^\pktlen, \nonumber\\
\vdots \nonumber\\
\Pi_s\oplus\cdots\oplus\Pi_1 \oplus \ker\Bs{F}_s = \Fbb_q^\pktlen \label{eq:Kernel_DirectSum}
\end{gather}
where ``$\oplus$'' is the direct sum of two disjoint subspaces. 
Moreover for $i\in [1:s]$ we have $\dim\Pi_i=\rank\Bs{F}_i - \rank\Bs{F}_{i-1}$. For more clarification, Figure~\ref{fig:Kernel_of_F} demonstrates the proposition.
\end{proposition}

\begin{figure}
\centering
\includegraphics[scale=0.8]{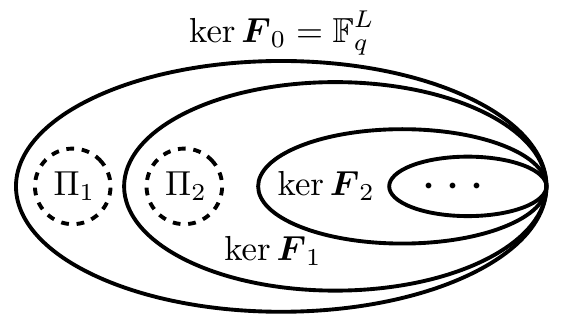}
\caption{Demonstration of Proposition~\ref{prop:Kernel_DirectSum_DetAchvScheme}. Here, we have $\vec{0}=\ker\Bs{F}_s\subset \ker\Bs{F}_{s-1} \subset \cdots\subset \ker\Bs{F}_0=\Fbb_q^\pktlen$ and $\Pi_1,\ldots,\Pi_s$ satisfy \eqref{eq:Kernel_DirectSum}.}
\label{fig:Kernel_of_F}
\end{figure}


In our proposed achievability scheme, Alice uses superposition coding
where she creates a vector
\begin{equation}\label{eq:SuperPosCoding}
{X}_\alice[t] = {X}_{\alice, 1}[t]+\cdots + {X}_{\alice, s}[t],
\end{equation}
such that ${X}_{\alice, i}[t]\in\Pi_i$. Because of \eqref{eq:Kernel_DirectSum},
$\{\Pi_1,\ldots,\Pi_s\}$ form a basis for $\Fbb_q^\pktlen$ so every vector
${X}_\alice[t] \in\Fbb_q^\pktlen$ can be uniquely decomposed as \eqref{eq:SuperPosCoding}.  
Now each ${X}_{\alice, i}[t] \in\Pi_i$ can be considered as a vector that
is transmitted by Alice and will be received independently 
by each trusted terminal or Eve with erasure probability 
$\theta_i = \sum_{j=0}^{i-1} \delta_i$.
Note that the vector ${X}_{\alice, i}[t]$ is 
correctly received by the $r$th receiver only if $S_r\ge i$.
 
So we may view the broadcast channel from Alice to the rest of terminals
as $s$ independent \emph{packet erasure channels}; where $\Pi_i$ is the set
of messages transmitted over the $i$th channel (layer) and the erasure
probability of the $i$th channel is $\theta_i$.

Then we can proceed as follows. On the $k$th layer, we run independently the scheme propose in \cite{siavoshani_group_2010,argyraki_creating_2013}
for the secret key sharing problem over an erasure broadcast channel (see also Appendix~\ref{sec:GrpSecKeyAgr_ErasureBrdcstChnl} for more details). Then, we can state the following result.

\begin{theorem}\label{thm:AchvAecrecyRate-DeterministicChnl}
The achievable SKG rate of the above scheme for each layer $k$ is given by
${R}^{\mathsf{det}}_k = (1-\theta_k)\theta_k \dim(\Pi_k)\log{q}$.
So for the total achievable secrecy rate we have
\begin{align*}
{R}^{\mathsf{det}}_s &= \sum_{i=1}^s \theta_i (1-\theta_i)  \dim{\left(\Pi_i\right)} \log{q} \nonumber\\
&= \sum_{i=1}^s \left[ \rank{\Bs{F}_i} - \rank{\Bs{F}_{i-1} }  \right] \theta_i (1-\theta_i) \log{q}.
\end{align*}
\end{theorem}

Observe that this matches the upper bound stated in Theorem~\ref{thm:SecKeyGenCapacity_UpperBound_Deterministic}, and therefore yields a characterization of the group key-agreement rate for deterministic channels, \ie, this completes the proof of Theorem~\ref{thm:GrpScrtKey-DetBrdcstChnl-MainRslt}. 

\begin{remark}
This result can be easily extended to the asymmetric case where the channels to the legitimate users are not statistically identical (but still independent). Moreover, notice that the key-generation rate is the same for any $m\geq 2$. This is in fact similar to the erasure channel case \cite{siavoshani_group_2010} (see also Appendix~\ref{sec:GrpSecKeyAgr_ErasureBrdcstChnl}),
where the critical difference between $m=2$ and $m>2$ is that the key-reconciliation necessitated the use of ideas from NC.
\end{remark}

\section{Group Secret Key Agreement over State-dependent Gaussian Broadcast Channels}\label{sec:GrpSecKeyAgr_GaussBrdcstChnl}
In this section, by using the results derived in the previous sections, we will study the secret key generation capacity among multiple terminals having access to a state-dependent Gaussian broadcast channel. We will derive upper and lower bounds for the secret key generation capacity. Although, the proposed bounds are not matched in general, we will show that they will match in the high-dynamic range, high-SNR regime in a degree of freedom sense.

\subsection{Upper Bound for the Key Generation Capacity}
In order to upper bound the secrecy capacity for the Gaussian broadcast channel, we cannot apply the result of 
Theorem~\ref{thm:SecrecyUpBound-CsNa08-forIndpChnl} directly because
this result has been derived under the assumption that the transmitted
and received symbols are discreet. However, the work in \cite{chan_multiterminal_2014} 
has extended the results of \cite{csiszar_secrecy_2008} for continuous channels.
So by using \cite[Theorem~3.2]{chan_multiterminal_2014}, we can write an 
upper bound for the secrecy capacity similar to Theorem~\ref{thm:SecrecyUpBound-CsNa08} with the addition of a power constraint over the transmitted symbols.
Then we can state the following result, as stated in Theorem~\ref{thm:SecKeyGenCapacity_UpperBound_Gaussian}.

\begin{theorem}\label{thm:SecKeyGenCapacity_UpperBound_Gaussian}
The key generation capacity of the Gaussian broadcast channel 
given in \eqref{eq:GaussianChannelModel}
using public discussions is upper bounded as follows
\begin{equation}
C^{\mathsf{gaus}}_s \le \frac{1}{2} L \sum_{i=0}^{s}\sum_{j=0}^{s} \delta_i\delta_j  \log \left(1 + \frac{h_i P_{\mathrm{max}}}{1 + h_j P_{\mathrm{max}}} \right).
\end{equation}
\end{theorem}

\begin{proof}
Using \cite[Theorem~3.2]{chan_multiterminal_2014} and by proceeding similar steps to the proof of Theorem~\ref{thm:SecrecyUpBound-CsNa08-forIndpChnl} (see Appendix~\ref{apn:SomeProofs}), we can write
\begin{align*}
C^{\mathsf{gaus}}_s &\le \min_{j\in[1:m-1]} \sup_{\begin{subarray}{c} P_{X_0}:\\ \Expc{\|X_0 \|^2} \le L P_{\mathrm{max}} \end{subarray}} I({X}_0; {X}_j| {X}_\eve) \nonumber \\
  & \stackrel{\text{(a)}}{=} \sup_{\begin{subarray}{c} P_{X_\alice}:\\ \Expc{\|X_\alice \|^2} \le L P_{\mathrm{max}} \end{subarray}} I({X}_\alice; {X}_\bob| {X}_\eve),
\end{align*}
where (a) is true because of the symmetry.
Hence, there exists an input distribution $P_{X_\alice}$ such that $\Expc{\|X_\alice \|^2} \le L P_{\mathrm{max}}$ where the secrecy capacity is upper bounded as follows
{\allowdisplaybreaks[4]
\begin{align*}
C^{\mathsf{gaus}}_s &\le I({X}_\alice; {X}_\bob| {X}_\eve) \nonumber\\
  &= I({X}_\alice; \hat{{X}}_\bob,S_\bob|\hat{{X}}_\eve,S_\eve) \nonumber\\
  &= H(\hat{X}_\bob,S_\bob|\hat{X}_\eve,S_\eve) - H(\hat{X}_\bob,S_\bob |\hat{X}_\eve,S_\eve , X_\alice)\nonumber\\
  &\stackrel{\text{(a)}} = H(\hat{X}_\bob,S_\bob|\hat{X}_\eve,S_\eve) - H(\hat{X}_\bob,S_\bob| X_\alice)\nonumber\\
  &= H(\hat{X}_\bob,S_\bob| \hat{X}_\eve,S_\eve) - H(S_\bob| X_\alice) - H(\hat{X}_\bob|S_\bob, X_\alice)\nonumber\\
  &\stackrel{\text{(b)}} = H(\hat{X}_\bob,S_\bob| \hat{X}_\eve,S_\eve) - H(S_\bob) - H(Z_\bob)\nonumber\\
  &= H(\hat{X}_\bob,\hat{X}_\eve|S_\eve,S_\bob) + H(S_\eve,S_\bob) - H(\hat{X}_\eve,S_\eve)-H(S_\bob) - H(Z_\bob)\nonumber\\
  &= H(\hat{X}_\bob,\hat{X}_\eve|S_\eve,S_\bob)-H(\hat{X}_\eve|S_\eve) - H(Z_\bob). 
\end{align*}}\hspace{-3pt}
where (a) is true since we have the Markov chain $X_\bob\leftrightarrow X_\alice \leftrightarrow X_\eve$ and
(b) follows from the fact that the state variables are independent of $X_\alice$ and given $X_\alice$ and $S_\bob$ the only uncertainty left in $\hat{X}_\bob$ is that of noise $Z_\bob$.
Now the above relation can be more simplified as follows
{\allowdisplaybreaks[4]
\begin{align}\label{eq:Maurer-GaussUpBound-1}
C^{\mathsf{gaus}}_s &\le \sum_{i=0}^{s}\sum_{j=0}^{s} \delta_i\delta_j H(\hat{X}_\bob,\hat{X}_\eve|S_\eve=j,S_\bob=i)  - \sum_{k=0}^{s} \delta_k H(\hat{X}_\eve|S_\eve=k) - H(Z_\bob)\nonumber\\
  &=\sum_{i=0}^{s}\sum_{j=0}^{s} \delta_i\delta_j H(\sqrt{h_i} X_\alice + Z_\bob, \sqrt{h_j} X_\alice + Z_\eve)  - \sum_{k=0}^{s} \delta_k H(\sqrt{h_k} X_\alice + Z_\eve) - H(Z_\bob)\nonumber\\
  &=\sum_{i=0}^{s}\sum_{j=0}^{s} \delta_i\delta_j H(\sqrt{h_i} X_\alice + Z_\bob| \sqrt{h_j} X_\alice + Z_\eve) - H(Z_\bob) \nonumber\\
  &\stackrel{\text{(a)}}{\leq} \sum_{i=0}^{s}\sum_{j=0}^{s} \frac{\delta_i\delta_j}{2} 
    \log \Big[ (2\pi e)^L \times  \det\Big( \cov(\sqrt{h_i} X_\alice + Z_\bob| \sqrt{h_j} X_\alice + Z_\eve) \Big) \Big] - H(Z_\bob),
\end{align}}\hspace{-5pt}
where (a) follows from the fact that for a fixed variance, Gaussian distribution maximizes the entropy.

The inequality (a) in \eqref{eq:Maurer-GaussUpBound-1} is achieved when 
$(\sqrt{h_i} X_\alice + Z_\bob| \sqrt{h_j} X_\alice + Z_\eve)$ 
has a Gaussian distribution. A sufficient condition for this to be satisfied 
is when $X_\alice$, $Z_\bob$, and $Z_\eve$ are Gaussian and independent, namely, $X_\alice \sim N(\vec{0},P_{\mathrm{max}}\Bs{I}_L)$, $Z_\bob \sim N(\vec{0},\Bs{I}_L)$, and $Z_\eve \sim N(\vec{0},\Bs{I}_L)$. 
This observation makes the calculation of 
\begin{equation*}
\frac{1}{2}\log \left[ (2{\pi}e)^L \det\left(\cov( \sqrt{h_i} X_\alice+ Z_\bob| \sqrt{h_j} X_\alice+ Z_\eve) \right) \right]
\end{equation*}
much easier as it is equivalent to the evaluation of 
$H( \sqrt{h_i} {X}_\alice + {Z}_\bob, \sqrt{h_j} {X}_\alice + {Z}_\eve) - H( \sqrt{h_j} {X}_\alice + {Z}_\eve)$ 
when ${X}_\alice$, ${Z}_\bob$, and ${Z}_\eve$ 
are Gaussian and independent as shown below,
{\allowdisplaybreaks[4]
\begin{align*}
&\frac{1}{2}\log \left[ (2{\pi}e)^L \det\left(\cov( \sqrt{h_i} {X}_\alice + {Z}_\bob| \sqrt{h_j} {X}_\alice + {Z}_\eve) \right) \right] = \nonumber\\
& \hspace{45pt} = H\big( \sqrt{h_i} {X}_\alice+ {Z}_\bob, \sqrt{h_j} {X}_\alice+ {Z}_\eve \big) - H\big( \sqrt{h_j} {X}_\alice+ {Z}_\eve \big) \nonumber\\
& \hspace{45pt} = \sum_{k=1}^L H\big( \sqrt{h_i} X_{\alice,k} + Z_{\bob,k} , \sqrt{h_j} X_{\alice,k} + Z_{\eve,k} \big) - H\big( \sqrt{h_j} X_{\alice,k} + Z_{\eve,k} \big) \nonumber\\
& \hspace{45pt} = \frac{L}{2} \Big[ \log \big( (2{\pi} e)^2(1 + h_i P_{\mathrm{max}} + h_j P_{\mathrm{max}}) \big)  - \log \big( 2\pi e(1+ h_j P_{\mathrm{max}}) \big) \Big],
\end{align*}}\hspace{-4pt}
where $\Expc{X_{\alice,k}^2}= P_{\mathrm{max}}$ and $\Expc{Z_{\bob,k}^2}=\Expc{Z_{\eve,k}^2}=1$ 
for all $k\in[1:L]$.

Hence, the upper bound on the secrecy capacity reads as follows
{\allowdisplaybreaks[4]
\begin{align*}
C^{\mathsf{gaus}}_s 
%
%
  & \le \sum_{i=0}^{s}\sum_{j=0}^{s} \frac{\delta_i\delta_j L}{2} \log \left[ (2{\pi} e)^2(1 + h_i P_{\mathrm{max}} + h_j P_{\mathrm{max}}) \right] -\sum_{i=0}^{s}\sum_{j=0}^{s} \frac{\delta_i\delta_j L}{2}  \log \left[ 2\pi e(1+ h_j P_{\mathrm{max}}) \right]  - \frac{L}{2} \log(2\pi e)\nonumber\\
%
  &  = \frac{1}{2} L \sum_{i=0}^{s}\sum_{j=0}^{s} \delta_i\delta_j  \log \left(1 + \frac{h_i P_{\mathrm{max}}}{1 + h_j P_{\mathrm{max}}} \right),
\end{align*}}\hspace{-5pt}
where we are done.
\end{proof}

\subsection{Lower Bound for the Key Generation Capacity}
\label{sec:GrpSecretKey-GaussBrdcstChnl-LowerBound}
Before stating our achievability scheme, let us first define a ``nested message set, degraded channel'' wiretap scenario.
\begin{definition}\label{def:ExtWiretapChannel}
Assume a wiretap channel scenario where there is a transmitter called Alice who broadcasts $X_\alice$ and there are $s+1$ receivers $Y_i$ where the $i$th receiver receives $Y_i$ according to the broadcast channel $(\mc{X}_\alice,p(y_0,\ldots,y_s|x),\mc{Y}_0\times\cdots\times\mc{Y}_s))$ such that 
\begin{equation*}
p(y_0,\ldots,y_s|x_\alice)= p(y_s|x_\alice)\cdot p(y_{s-1}|y_s)\cdots p(y_0|y_1).
\end{equation*}

Suppose that Alice has $s$ messages $W_1,\ldots,W_s$ where $W_i\in\{1,\ldots,2^{LR_i}\}$ and 
$W_i\sim\mathsf{Uni}([1:2^{LR_i}])$.
The goal is that she wants to broadcast these messages such that $\forall i$:\\
\emph{\textsf{(i)}} each message $W_i$ should be decodable by the receivers $Y_i,\ldots,Y_s$ 
with a negligible error probability, and\\
\emph{\textsf{(ii)}} all the receivers $Y_0,\ldots,Y_{i-1}$ should be ignorant about the message $W_i$, 
namely for the leakage rate we should have
\begin{equation}\label{eq:DefLeakageRate_ExtWiretap}
R_{\mathrm{leak},i}^{(L)}\triangleq\frac{1}{L}I(W_{i+1},\ldots,W_s;Y_i^{1:L})\le\epsilon_L, \forall i\in [0:s].
\end{equation}
\end{definition}

Now suppose that a multi-receiver wiretap scenario as defined in Definition~\ref{def:ExtWiretapChannel} consists of $s + 1$ independent
Gaussian channels where the $r$th channel is defined as follows
\begin{equation}\label{eq:GaussianWiretapChannelModel}
Y_r[t] = \sqrt{h_r} X[t] + Z_r[t],\quad\quad \forall r\in [0 : s],
\end{equation}
where $Z_r[t] \sim N(0, 1)$ and $h_r$ are some fixed constant representing the channel gains such that $h_0 < \cdots < h_s$.
We also assume that the channel input is subject to an average power constraint $P_{\mathrm{max}}$, \ie, $\frac{1}{L} \sum_{t=1}^L \Expc{X^2[t]} \le P_{\mathrm{max}}$.
Then we can state the following result.
\begin{theorem}\label{thm:GaussianWiretapAchievability}
Using a properly designed layered wiretap code similar to \cite{liang_broadcast_2009,liang_broadcast_2014}, we can achieve the following set of rates for the ``nested message set, degraded Gaussian wiretap channel,''
\begin{align}\label{eq:Layer_i_Rate}
R_i &= \frac{1}{2} \left[ \log\left(1+ \frac{h_i P_i}{1+ h_i I_i} \right) -\log\left(1+ \frac{h_{i-1} P_i}{1+ h_{i-1} I_i} \right) \right],
\end{align}
$\forall i\in[1:s]$, where $I_i \triangleq \sum_{j=i+1}^s P_j$.
\end{theorem}

\begin{proof}
In the following, we will describe a code construction that achieves the rates stated in Theorem~\ref{thm:GaussianWiretapAchievability}. Because it is very similar to \cite{liang_broadcast_2009} and also due to space limit, we only present a sketch of the proof for the theorem.

Assume that the code has $s$ layers that correspond to each channel where they are indexed from
$1$ up to $s$ (the channel $Y_0$ should decode nothing). To each layer a power constraint $P_i$ is assigned such
that $\sum_{i=1}^s P_i \le P_{\mathrm{max}}$. The transmitter uses superposition coding to encode each message $W_i$ that corresponds 
to layer $i$; namely, it broadcasts
\begin{equation*}
X[t]=\sum_{i=1}^s X_i[t],
\end{equation*}
over the channel described by \eqref{eq:GaussianWiretapChannelModel}. 
Then by receiving $Y_r$, the $r$th receiver uses successive decoding, that starts 
from the layer $1$ to decode $X_1$ assuming the rest of the layers as noise and 
subtracting $X_1$ from the received vector after decoding. Then it continues this process
to decode the rest of layers.

More precisely we construct $s$ codebooks $\mc{\hat{C}}_i(2^{L\hat{R}_i},L)$ 
each contains $2^{L\hat{R}_i}$ codewords $X_i^L$ of length $L$ by choosing in total $L 2^{L\hat{R}_i}$ symbols independently from the Gaussian distribution ${N}(0,P_i)$ where
\begin{equation*}
\hat{R}_i = \frac{1}{2}\log\left(1+ \frac{h_i P_i}{1+ h_i I_i} \right),
\end{equation*}
and $I_i=\sum_{j=i+1}^s P_j$.
Each codebook $\hat{\mc{C}}_i$, $0<i\le s$, is divided into $2^{LR_i}$ bins where
\begin{align*}
R_i &= \frac{1}{2} \left[ \log\left(1+ \frac{h_i P_i}{1+ h_i I_i} \right) -\log\left(1+ \frac{h_{i-1} P_i}{1+ h_{i-1} I_i} \right) \right].
\end{align*}
At every layer $i$, each message is mapped into one bin, and one codeword in the bin is randomly chosen. So, layer $i$ can transmits $2^{LR_i}$ messages. Following a similar argument as stated in \cite{liang_broadcast_2009,liang_broadcast_2014}, it can be shown that the above codebook satisfies the requirement of Definition~\ref{def:ExtWiretapChannel}.
%
%
\end{proof}

\begin{remark}
Note that all of the above discussions are also valid for complex channels. The only difference is that there will be no $\frac{1}{2}$ coefficient before rates given by \eqref{eq:Layer_i_Rate} and other expressions should be updated accordingly.
\end{remark}

Now, as described in the proof of Theorem~\ref{thm:GaussianWiretapAchievability}, by using a properly designed layered coding for the nested message set, degraded channel wiretap scenario, we can convert the Gaussian channel given in \eqref{eq:GaussianChannelModel} to a set of $s$ independent erasure channels where the erasure of the messages for each channel (layer) depends on the receiver channel state.
In fact using the layered coding scheme for the wiretap channel, we mimic the orthogonality behaviour that we have for the deterministic channel as described by \eqref{eq:TransMatrixNested_Structure_1} and \eqref{eq:Kernel_DirectSum}.

To be more specific, we assume that Alice broadcasts an $\pktlen$-length vector
\begin{equation*}
{X}_\alice[t] = \sum_{i=1}^s {X}_{\alice, i}[t],
\end{equation*}
where she maps $W_i$ (the messages corresponding to the $i$th layer) to ${X}_{\alice, i}[t]$ according to the codebook described in the proof of Theorem~\ref{thm:GaussianWiretapAchievability}.
%
From the proof we know that the receiver $r$ which observes the channel state $S_r=i$ can decode messages up to layer$i$ and is ignorant about messages of layers above $i$. So, equivalently, we can say that the message $W_i$ experiences erasure probability 
$\theta_i=\sum_{j=0}^{i-1} \delta_j$,
when it passes through the channel \eqref{eq:GaussianChannelModel}.

Now for each layer $i$, we run the interactive secret key sharing scheme introduced in 
\cite{siavoshani_group_2010,argyraki_creating_2013} (see also see Appendix~\ref{sec:GrpSecKeyAgr_ErasureBrdcstChnl})
where Alice broadcasts an $n$-length sequence of random messages, \ie, $W_i^n$. Then, by discussing over the public channel, the trusted terminals reconcile their secret messages to build a common key. The key generation rate for each layer is $\Delta_i L R_i$, so for a fixed power allocation we achieve the following secrecy rate
\begin{equation*}
{R}^{\mathsf{gaus}}_s \le \sum_{i=1}^s \Delta_i L R_i,
\end{equation*}
where $R_i$ is defined in \eqref{eq:Layer_i_Rate} and
$\Delta_i = (1-\theta_i) \theta_i$.

The maximum secrecy rate is obtained by optimizing the above rate over the power allocations $\{P_i\}_{i=1}^s$. Thus we can write
\begin{equation}\label{eq:PowerAloocation-P1}
{R}^{\mathsf{gaus}}_s = \left\{\begin{array}{ll}
\max & \sum_{i=1}^s \Delta_i L R_i \\
\text{subject to} & \sum_{i=1}^s P_i\le P_{\mathrm{max}}\\
 & P_i\ge 0,\quad \forall i\in[1:s].
\end{array} \right.
\end{equation}
Because $R_1$ is an increasing function of $P_1$ when other $P_i$ are kept fixed and $R_i$ does not depend on $P_1$
for $i>1$ we can write the power constant inequality as an equality. We also apply a change of variables to Problem~\ref{eq:PowerAloocation-P1} from $\{P_i\}$ to $\{I_k\}$. So we can rewrite~\eqref{eq:PowerAloocation-P1} in the canonical form (see \cite{boyd_convex_2004}) as follows
\begin{equation}\label{eq:PowerAloocation-P2}
{R}^{\mathsf{gaus}}_s = \left\{\begin{array}{ll}
\min & -\sum_{i=1}^s \Delta_i L R_i \\
\text{subject to} & -[I_{k-1}-I_k]\le 0,\ \forall k\in[1:s],
\end{array} \right.
\end{equation}
where for convenience we define $I_0\triangleq P_{\mathrm{max}}$, $I_s\triangleq 0$, and we have also
\begin{equation*}
R_i = \frac{1}{2} \log\left(\frac{1+h_i I_{i-1}}{1+ h_i I_i} \cdot \frac{1+h_{i-1} I_i}{1+ h_{i-1} I_{i-1}} \right).
\end{equation*}
In \S\ref{sec:GrpSecretKey-GaussBrdcstChnl-SolvePowerAllocProblem}, we will focus on solving the optimization problem~\eqref{eq:PowerAloocation-P2}.

\section{Solving the Non-convex Power Allocation Problem}
\label{sec:GrpSecretKey-GaussBrdcstChnl-SolvePowerAllocProblem}
Here, we present how the optimization~problem~\eqref{eq:PowerAloocation-P2} can be solved. Our final result is not in closed form but instead we propose a recursive algorithm (\ie, a dynamic program) that finds all the possible solutions of KKT\footnote{Karush--Kuhn--Tucker conditions (\eg, see \cite{boyd_convex_2004}).} conditions (which provide necessary conditions for an optimal solution to the optimization~problem~\eqref{eq:PowerAloocation-P2}) and find an optimum solution by searching among them. By using the proposed algorithm, we reduce the search space of the optimization~problem~\eqref{eq:PowerAloocation-P2} from a multi-dimensional continuous space to a finite elements set; \ie, the set of solutions to the KKT conditions. In this sense, the final result is exact (the proposed algorithm is not a numerical approximation), but it is hard to describe the solution in a single closed form equation for all possible parameters involved in the secrecy problem (\eg, channel gains, probability distribution over states, etc.). However, note that for each set of given problem parameters, it is possible to state the final solution in terms of these given parameters (but here we only focus on deriving the final ``value'' of the solution, not its ``expression''). To find the optimal solutions of the above-mentioned optimization problem, we proceed as follows.

Because the constraints of optimization problem~\eqref{eq:PowerAloocation-P2} are affine, we can use the KKT conditions to derive a set of necessary conditions for the optimum power allocation (\eg, see \cite[Chapter~5]{bertsekas_convex_2003}). By defining the Lagrangian $\mf{L}$ as
\begin{equation*}
\mf{L}(P_1,\ldots,P_s,\lambda_1,\ldots,\lambda_s) = -\sum_{i=1}^s \Delta_i L R_i + \sum_{i=1}^s \lambda_i [I_i-I_{i-1}],
\end{equation*}
and applying the KKT theorem (\eg, see \cite[Chapter~5]{boyd_convex_2004}), we write a set of necessary conditions for the optimal solution of \eqref{eq:PowerAloocation-P2} as follows
\begin{equation}\label{eq:PowerAlocation-KKTNecessaryCond}
\left\{\begin{array}{ll}
\frac{\partial \mf{L}}{\partial I_k} = 0, & \forall k\in[1:s-1],\\
\lambda_k [I_k-I_{k-1}] = 0, & \forall k\in[1:s],\\
I_k \le I_{k-1}, & \forall k\in[1:s],\\
\lambda_k \ge 0, & \forall k\in[1:s].        
\end{array} \right.
\end{equation}
By taking the derivative of $\mf{L}$ with respect to $I_k$, $\forall k\in[1:s-1]$, and doing some algebra we get
\begin{align}\label{eq:PowerAllocation_LagrangianDerivative_LinearCase}
0=\frac{\partial \mf{L}}{\partial I_k} &= \frac{({h}_{k+1} \beta_k-{h}_{k-1}\alpha_k)I_k-(\alpha_k-\beta_k)}{(\ln{2})(1+ {h}_{k-1} I_k)(1+ {h}_{k} I_k)(1+ {h}_{k+1} I_k)} + (\lambda_k - \lambda_{k+1}) \nonumber\\
 &= F^{(1)}_k(I_k) + (\lambda_k - \lambda_{k+1}),
\end{align}
where $\alpha_k\triangleq ({h}_{k+1}- {h}_k)\Delta_{k+1}$, $\beta_k\triangleq ({h}_k - {h}_{k-1})\Delta_{k}$ and $F^{(1)}_k(I_k)$ is defined accordingly. Notice that because $I_k$'s are positive variables the denominator of $F^{(1)}_k(I_k)$ is strictly positive. For the ease of reference, some of the important variables of our problem are gathered in Table~\ref{tab:Important_Variables}.

\begin{table}
\centering
\begin{tabular}{| c | c |}
\hline
\textbf{Variable} & \textbf{Definition}\\
\hline
$\theta_k\ (\forall k\in[1:s])$ & \parbox[t]{\TableWidth}{The effective erasure probability that the message $W_k$ (message of $k$th layer) experience which is  $\sum_{i=0}^{k-1} \delta_i$.}\\
\hline
$\Delta_k\ (\forall k\in[1:s])$ & \parbox[t]{\TableWidth}{A dummy variable which is $\theta_k (1-\theta_k)$.}\\
\hline
$I_k\ (\forall k\in[0:s])$ & \parbox[t]{\TableWidth}{The contribution of the interference of all layers above $k$ to the decoding of the $k$th layer which is $\sum_{i=k+1}^s P_i$. For convenience we define $I_0=P_{\mathrm{max}}$ and $I_s=0$.}\\ 
\hline
$h_k\ (\forall k\in[0:s])$ & \parbox[t]{\TableWidth}{The square of channel gains. Remember that $h_0 < \cdots < h_s$.}\\
\hline
$\lambda_k\ (\forall k\in[1:s])$ & \parbox[t]{\TableWidth}{The Lagrangian multipliers of optimization problem \eqref{eq:PowerAloocation-P2}.}\\
\hline
$\alpha_k\ (\forall k\in[1:s-1])$ &  \parbox[t]{\TableWidth}{$=({h}_{k+1}- {h}_k)\Delta_{k+1}$.}\\
\hline
$\beta_k\ (\forall k\in[1:s-1])$ & \parbox[t]{\TableWidth}{$=({h}_k - {h}_{k-1})\Delta_{k}$.}\\
\hline
$\{I^*_k\}_{k=0}^s$ & \parbox[t]{\TableWidth}{With an abuse of notation denotes any solution to the set of KKT conditions stated in \eqref{eq:PowerAlocation-KKTNecessaryCond}.}\\
\hline
$\{I^{**}_k\}_{k=0}^s$ & \parbox[t]{\TableWidth}{The optimum power allocation of the optimization problem \eqref{eq:PowerAloocation-P2} that also satisfies \eqref{eq:PowerAlocation-KKTNecessaryCond}.}\\
\hline
$r^{(1)}_k \ (\forall k\in[0:s])$ & \parbox[t]{\TableWidth}{The root of the numerator of $F^{(1)}(x)$ defined in \eqref{eq:PowerAllocation_LinearCase_root}. Note that we define $r^{(1)}_0\triangleq P_{\mathrm{max}}$ and $r^{(1)}_s \triangleq 0$.}\\
\hline
$r^{(2)}_{k,1}, r^{(2)}_{k,2}$ & \parbox[t]{\TableWidth}{The real roots (if exist) of the numerator of $F^{(2)}(x)$ defined in \eqref{eq:PowerAllocation_LagrangianDerivative_QuadraticCase}.}\\
\hline
\end{tabular}
\caption{Explanation for some of the important variables.}\label{tab:Important_Variables}
\end{table}

The main idea of our proof is to propose a recursive algorithm that first finds all the solutions of the KKT equations \eqref{eq:PowerAlocation-KKTNecessaryCond}, (with an abuse of notation) each is denoted by $\{I^*_k\}_{k=0}^s$. Then among these solutions finds the one that maximizes the secrecy rate given by \eqref{eq:PowerAloocation-P2}, which is denoted by $\{I^{**}_k\}_{k=0}^s$.

To this end, in every iteration, the proposed algorithm picks some $k$ (for that $I^*_k$ is not determined yet)  and then determine the sign of $F^{(1)}_k(x)$ (or as will be discussed later, for some cases determine the sign of $F^{(2)}_k(x)$ which will be defined in \eqref{eq:PowerAllocation_LagrangianDerivative_QuadraticCase}). Then using \eqref{eq:PowerAllocation_LagrangianDerivative_LinearCase} (or in some cases using \eqref{eq:PowerAllocation_LagrangianDerivative_QuadraticCase}) in addition to the complementary slackness condition, it determines whether we have to examine the following three cases: (i) $I^*_k=I^*_{k-1}$,  (ii) $I^*_k=I^*_{k+1}$, or (iii)  the value of $I^*_k$ is determined in this iteration. Hence, at the end of each iteration the size of the optimization problem is reduced by one (either $I^*_k$ is determined in this iteration or is equal to the $I^*_{k-1}$ or $I^*_{k+1}$). Ignoring the details for a moment, we can repeat the above procedure until all values of $I^*_k$ are determined. These sets of $I^*_k$'s are the solutions to the KKT conditions \eqref{eq:PowerAlocation-KKTNecessaryCond}.

Considering more details, we can proceed as follows. First, let $r^{(1)}_k$ to be the root of the numerator of $F^{(1)}_k(x)$ (assuming ${h}_{k+1}\beta_k- {h}_{k-1}\alpha_k\neq 0$), namely,
\begin{align}\label{eq:PowerAllocation_LinearCase_root}
r^{(1)}_k &\triangleq  \frac{\alpha_k-\beta_k}{{h}_{k+1}\beta_k- {h}_{k-1}\alpha_k},\quad\quad \forall k\in [1:s-1],
\end{align}
and by convention set $r^{(1)}_0 \triangleq P_{\mathrm{max}}$ and $r^{(1)}_s\triangleq 0$. 
Then we can observe the following different situations, as stated in Case~\ref{cond:PowerAloc_KKT_LinearCase_Conds}.

\begin{case}[Linear Case\footnote{Here by the linear case, we mean that the numerator of $F^{(1)}_k(I_k)$ in \eqref{eq:PowerAllocation_LagrangianDerivative_LinearCase} is a linear function of $I_k$.}]
\label{cond:PowerAloc_KKT_LinearCase_Conds}
Based on different values of problem parameters, for each $k\in[1:s-1]$, we have various cases for a solution $I^*_k$ that satisfies \eqref{eq:PowerAlocation-KKTNecessaryCond} as follows:\\
(1)	If ${h}_{k+1}\beta_k- {h}_{k-1}\alpha_k = 0$ (which means that the numerator of $F^{(1)}_k(x)$ is a constant and $r^{(1)}_k$ is not defined), then because of the ordering over channel gains we should have $\alpha_k > \beta_k$. Now, because of \eqref{eq:PowerAllocation_LagrangianDerivative_LinearCase} and since $\lambda_k$'s are non-negative, we should have $ \lambda_{k}^* > 0$ which by using the complementary slackness condition leads to $ I_k^*=I_{k-1}^*$. This is equivalent to $P^*_{k}=0$.\\
(2)	If $\alpha_k<\beta_k$ then we get $h_{k+1}\beta_k > h_{k-1}\alpha_k$, so we have $r^{(1)}_k < 0$ and $F^{(1)}_k(x)>0$ for $x\ge 0$. Because of \eqref{eq:PowerAllocation_LagrangianDerivative_LinearCase}, we conclude that $\lambda_{k+1}^*>0$ which by using the complementary slackness condition results in $I_k^*=I_{k+1}^*$, \ie,  $P^*_{k+1}=0$.\\
(3)	If ${h}_{k+1}\beta_k < {h}_{k-1}\alpha_k$ we can conclude that $\alpha_k>\beta_k$ so we have $r^{(1)}_k < 0$ and $F^{(1)}_k(x)<0$ for $x\ge 0$. Because of \eqref{eq:PowerAllocation_LagrangianDerivative_LinearCase}, we conclude that $\lambda_k^*>0$ which by using the complementary slackness condition results in $I_k^*=I_{k-1}^*$, \ie, $P^*_k=0$.\\
(4)	If $\alpha_k>\beta_k$ and ${h}_{k+1}\beta_k > {h}_{k-1}\alpha_k$ then we have $r^{(1)}_k >0$. Moreover, we have $F^{(1)}_k(x)>0$ for $x> r^{(1)}_k$ and $F^{(1)}_k(x)<0$ for $x < r^{(1)}_k$. Now, there exists the following different cases:
\begin{itemize}
		\item[(a)] If $P_{\mathrm{max}} < r^{(1)}_k$ then we have $F^{(1)}_k(x)<0$ for $x\le P_{\mathrm{max}}$. From \eqref{eq:PowerAllocation_LagrangianDerivative_LinearCase}, we conclude that $\lambda_k^*>0$ which leads to $I_k^*=I_{k-1}^*$, \ie, $P_k=0$.
		\item[(b)] If $0 < r^{(1)}_k \le P_{\mathrm{max}}$ then we have $F^{(1)}_k(x)<0$ for $x < r^{(1)}_k$ and $F^{(1)}_k(x)>0$ for $x > r^{(1)}_k$. Now $I^*_k$ can be equal to $r^{(1)}_k$ without any further requirement. However, if we have $r^{(1)}_k < I^*_k$ then we should have $I^*_k=I^*_{k+1}$. Similarly if we have $I^*_k < r^{(1)}_k$ then we have to have $I^*_k=I^*_{k-1}$.
\end{itemize}
\end{case}

The above different cases are derived under the assumption that we do not have any extra information about a solution $\{I^*_k\}$. In particular, prior to solving the KKT conditions, we do not know whether we have $I^*_k=I^*_{k\pm l}$ (for some valid $k$ and $l$) or not. Note that due to the ordering on the optimal solution $I^*_k$'s imposed by \eqref{eq:PowerAloocation-P2}, if we have $I^*_k=I^*_{k+l}$ then we should also have $I^*_{k}=I^*_{k+1}=\cdots=I^*_{k+l}$.

Now suppose that, by some mean (\eg, from the previous iterations of our proposed algorithm for finding the solutions of KKT equations \eqref{eq:PowerAlocation-KKTNecessaryCond}), we know that $I^*_k=I^*_{k+l}$. This knowledge enables us to reduce the size of the optimization problem \eqref{eq:PowerAloocation-P2}. Notice that after having this information, the derivative of the Lagrangian $\mf{L}$ with respect to $I_k$ 
is not given by \eqref{eq:PowerAllocation_LagrangianDerivative_LinearCase} anymore. More precisely, let us assume that $I^*_k=I^*_{k+l}$. Then, taking the derivative of $\mf{L}$ with respect to $I_k$ and by doing some algebra, we can write 
\begin{align}\label{eq:PowerAllocation_LagrangianDerivative_QuadraticCase}
 0 = \frac{\partial \mf{L}}{\partial I_k} &= \frac{\Delta_k(1+h_{k+l}I_k)(1+h_{k+l+1}I_k)(h_k-h_{k-1})}{(\ln{2})(1+h_{k-1}I_k)(1+h_{k}I_k)(1+h_{k+l}I_k)(1+h_{k+l+1}I_k)} \nonumber\\
    &\quad - \frac{\Delta_{k+l+1}(1+h_{k-1}I_k)(1+h_{k}I_k)(h_{k+l+1}-h_{k+l})}{(\ln{2})(1+h_{k-1}I_k)(1+h_{k}I_k)(1+h_{k+l}I_k)(1+h_{k+l+1}I_k)} + (\lambda_k-\lambda_{k+l+1})  \nonumber\\
    &= F^{(2)}_k(I_k) + (\lambda_k-\lambda_{k+l+1}).
\end{align}
Notice that for all values of $l\in [1:s-k]$, the numerator of $F^{(2)}_k(I_k)$ is a quadratic function in $I_k$ and the denominator is strictly positive for $I_k\ge 0$ because $h_k$'s are positive real quantities.

Similar to the Case~\ref{cond:PowerAloc_KKT_LinearCase_Conds}, here in this case, we can also find the real roots $r^{(2)}_{k,1}$ and $r^{(2)}_{k,2}$ of the numerator of $F^{(2)}_k(x)$ and find the sign of the function $F^{(2)}_k(x)$ for different values of $x\in[0,P_{\mathrm{max}}]$ based on the place of these roots. Hence, based on the real roots of the numerator of $F^{(2)}_k(x)$, we can write a set of different conditions similar to Case~\ref{cond:PowerAloc_KKT_LinearCase_Conds}, as stated in Case~\ref{cond:PowerAloc_KKT_QuadraticCase_Conds}. 
\begin{case}[Quadratic Case\footnote{Here by the quadratic case, we mean that the numerator of $F^{(2)}_k(I_k)$ in \eqref{eq:PowerAllocation_LagrangianDerivative_QuadraticCase} is a quadratic function of $I_k$.}]
\label{cond:PowerAloc_KKT_QuadraticCase_Conds}
Here, we do not write all the different possibilities for \eqref{eq:PowerAllocation_LagrangianDerivative_QuadraticCase} because the idea is very similar to Case~\ref{cond:PowerAloc_KKT_LinearCase_Conds}.
Instead, we explain the main part of the procedure in the following. In general, based on the position of the roots of numerator of $F^{(2)}_k(x)$, we can potentially have up to five different cases. To clarify the method, for example, consider the situation where the numerator of $F^{(2)}_k(x)$ has two distinct real roots $r^{(2)}_{k,1}<r^{(2)}_{k,2}$ where $r^{(2)}_{k,1}, r^{(2)}_{k,2} \in [0,P_{\mathrm{max}}]$. Then we have to consider the following five different cases for the solution $I^*_k$: (1) $I^*_k \in \left[0,r^{(2)}_{k,1}\right)$, (2) $I^*_k = r^{(2)}_{k,1}$, (3) $I^*_k\in \left(r^{(2)}_{k,1},r^{(2)}_{k,2}\right)$, (4) $I^*_k = r^{(2)}_{k,2}$, and (5) $I^*_k\in \left(r^{(2)}_{k,2},P_{\mathrm{max}} \right]$. In the items (1), (3), and (5) one can find the sign of $F^{(2)}_k(x)$ in the corresponding interval and based on that determine whether we should have $I^*_k=I^*_{k-1}$ or $I^*_{k}=I^*_{k+l+1}$.
\end{case}

\begin{remark}
It is worth to emphasize that the structure of the optimization problem \eqref{eq:PowerAloocation-P2} is such that the denominator of $F^{(1)}_k(x)$ and $F^{(2)}_k(x)$ are always strictly positive for $x\ge 0$. Moreover, the numerator of $F^{(1)}_k(x)$ is always at most a linear function of $x$ and that of $F^{(2)}_k(x)$ is always at most a quadratic function in $x$. This fact significantly simplifies finding the solutions of KKT equations \eqref{eq:PowerAlocation-KKTNecessaryCond} and hence solving the optimization problem \eqref{eq:PowerAloocation-P2}.
\end{remark}

The above discussion for different possible cases based on the numerator roots of $F^{(1)}_k$ and $F^{(2)}_k$ can be applied to any specific instance of the optimization problem \eqref{eq:PowerAloocation-P2}. As briefly explained before, the main idea is to apply a recursive algorithm that finds the set of solutions of \eqref{eq:PowerAlocation-KKTNecessaryCond} iteratively by determining variables $I_k$'s one by one. This procedure can also be considered as reducing the size of the optimization problem \eqref{eq:PowerAloocation-P2} (\ie, number of states) by one in every iteration.

To better explain our proposed method, we present the pseudo code of our algorithm in  Algorithm~\ref{alg:Solve_KKT_Conditions} and Algorithm~\ref{alg:Solve_KKT_RecursionMainPart}. 
For more clarification, in addition to the comments inside the pseudo code, the important variables of the pseudo code are explained separately in Table~\ref{tab:DescriptionImptVariables_in_Algs}.

Putting it together, we can describe our algorithm as follows (see also Algorithm~\ref{alg:Solve_KKT_Conditions} and Algorithm~\ref{alg:Solve_KKT_RecursionMainPart}). 
First, Algorithm~\ref{alg:Solve_KKT_Conditions} initializes a data structure $\mathrm{d}(i)$ for each $i\in[0:s]$ which contains the required information about $I^*_i$ (see Lines~\ref{alg:line:Init_begin} to \ref{alg:line:Init_end} of Algorithm~\ref{alg:Solve_KKT_Conditions}). Then it calls Algorithm~\ref{alg:Solve_KKT_RecursionMainPart} that is a recursive function.

Starting from the original KKT conditions, at every iteration, Algorithm~\ref{alg:Solve_KKT_RecursionMainPart} picks a number $k$ (such that $I_k$ is not determined yet) and apply Case~\ref{cond:PowerAloc_KKT_LinearCase_Conds} or Case~\ref{cond:PowerAloc_KKT_QuadraticCase_Conds} (depending if it is linear or quadratic case) to that particular $I_k$ (see Lines~\ref{alg:line:Pick_a_k}, \ref{alg:line:Linear_Case} and \ref{alg:line:Quadratic_Case} of Algorithm~\ref{alg:Solve_KKT_RecursionMainPart}). By doing so, the size of the original KKT conditions (\ie, number of undetermined variables $I_k$) is reduced by one and we may have up to five (in fact up to three if Case~\ref{cond:PowerAloc_KKT_LinearCase_Conds} holds and up to five if Case~\ref{cond:PowerAloc_KKT_QuadraticCase_Conds} holds) new set of KKT conditions to be solved.

Now, we can repeat the above process on each of these new set of conditions and go forward iteratively (see Lines~\ref{alg:line:Call_Recrsion_LinCase} and \ref{alg:line:Call_Recrsion_QuadCase} of Algorithm~\ref{alg:Solve_KKT_RecursionMainPart}). This procedure is like discovering a tree starting from some point as root (the root is determined by the first $k\in[1:s-1]$ picked up by the algorithm).  Note that many of these new set of conditions do not lead to valid solutions that satisfy the original KKT conditions \eqref{eq:PowerAlocation-KKTNecessaryCond}. This will be determined later as the algorithm proceeds by observing some contradictions on the intervals of $I_k$'s. This process continues until we obtain problems of zero size (that have all variable $I_k$'s determined and satisfy the original KKT conditions \eqref{eq:PowerAlocation-KKTNecessaryCond}; Line~\ref{alg:line:Dtermined_and_Satisfied} of Algorithm~\ref{alg:Solve_KKT_RecursionMainPart}) or at some point in the middle of the algorithm the determined $I_k$'s up to that point violate the KKT conditions (so this particular branch will be discarded; Lines~\ref{alg:line:Chk_Consistency_LinCase} and \ref{alg:line:Chk_Consistency_QuadCase} of Algorithm~\ref{alg:Solve_KKT_RecursionMainPart}).

The above-mentioned algorithm enables us to find all of the solutions to KKT conditions \eqref{eq:PowerAlocation-KKTNecessaryCond}. Then because the KKT equations provide a necessary condition on the optimal solution, it is sufficient to check among all of the solutions of KKT equations to find the optimal power allocation for the optimization problem~\eqref{eq:PowerAloocation-P2} (Line~\ref{alg:line:Find_Optimal_Power_Allocation} of Algorithm~\ref{alg:Solve_KKT_RecursionMainPart}). Consequently, the search space of the original optimization problem is reduced from the continuous space $\mbb{R}^{s-1}$ to a set of size at most $5^{s-1}$ elements. Note that this is the worst case analysis and in practice the size of the set can be much smaller than this number\footnote{In the examples discussed in the following, the number of solutions is something like $2$ or $3$.}.

\begin{algorithm}[H]
\caption{Finding all of the solutions that satisfy KKT conditions \eqref{eq:PowerAlocation-KKTNecessaryCond}.} \label{alg:Solve_KKT_Conditions}
\begin{algorithmic}[1]
  \Require $s$, $\{h_i\}_{i=0}^s$, $\{\Delta_k\}_{k=1}^s$, $P_{\mathrm{max}}$

  \ForAll{$i\in [0:s]$} \Comment{Initialization} \label{alg:line:Init_begin}
    \State $\mathrm{d}(i).\mathrm{l} = \mathrm{d}(i).\mathrm{u} = i$ \Comment{In general we may have $I^*_{\mathrm{d}(i).\mathrm{l}}=\cdots=I^*_{\mathrm{d}(i).\mathrm{u}}$}
    \State $\mathrm{d}(i).\mathrm{min} = 0$ \Comment{Initializing the lower bound on $I^*_i$}
    \State $\mathrm{d}(i).\mathrm{max} = P_{\mathrm{max}}$ \Comment{Initializing the upper bound on $I^*_i$}
    \State $\mathrm{d}(i).\mathrm{determined} = \mathrm{false}$ \Comment{At the beginning, the value of $I^*_i$ is not determined}
  \EndFor
  \State $\mathrm{d}(0).\mathrm{min} = P_{\mathrm{max}}$; \quad $\mathrm{d}(0).\mathrm{determined} = \mathrm{true}$ \Comment{Also part of the initialization}
  \State $\mathrm{d}(s).\mathrm{max} = 0$; \quad $\mathrm{d}(s).\mathrm{determined} = \mathrm{true}$ \Comment{Also part of the initialization} \label{alg:line:Init_end}

  \State $\mathrm{SolSet}=$ \Call{Recursion}{$\mathrm{d}$, $\{h_i\}_{i=0}^s$, $\{\Delta_k\}_{k=1}^s$}
  \State For all $I\in\mathrm{SolSet}$ find the one which maximizes the achievable rate $R$; call it $\mathrm{I}^{**}$ \label{alg:line:Find_Optimal_Power_Allocation}
  \State \Return $\mathrm{I}^{**}$
  \State \textbf{end}
\end{algorithmic}
\end{algorithm}

\begin{table}
\centering
\begin{tabular}{|c|c|}
\hline
\textbf{Variable} & \textbf{Description}\\
\hline
$\mathrm{d}(i)$ &  \parbox[t]{\TableWidth}{An array that contains the available information about the solution $I^*$ of \eqref{eq:PowerAlocation-KKTNecessaryCond} at every step of the algorithm. At the beginning, we have $i\in [0,s]$. However, the size of the problem becomes smaller in every iteration.}\\
\hline
$\mathrm{d}(i).\mathrm{l}$ and $\mathrm{d}(i).\mathrm{u}$ & \parbox[t]{\TableWidth}{Lower and Upper bounds on the indices of states such that we have $I^*_{\mathrm{d}(i).\mathrm{l}}=\cdots=I^*_{\mathrm{d}(i).\mathrm{u}}$.}\\
\hline
$\mathrm{d}(i).\mathrm{min}$ and $\mathrm{d}(i).\mathrm{max}$ & \parbox[t]{\TableWidth}{Determine the interval that $I_k^*$'s belongs to, \ie, $I_k^* \in \big[ \mathrm{d}(i).\mathrm{min},\mathrm{d}(i).\mathrm{max} \big]$ where $k \in \big[ \mathrm{d}(i).\mathrm{l} : \mathrm{d}(i).\mathrm{u} \big]$.}\\
\hline
$\mathrm{d}(i).\mathrm{determined}$ &  \parbox[t]{\TableWidth}{Let $k=\mathrm{d}(i).\mathrm{l}$. If the value of $I_k^*$ is completely determined, we have $\mathrm{d}(i).\mathrm{determined}=$~``$\mathrm{true}$'' otherwise it is equal to ``$\mathrm{false}$''.}\\
\hline
$\mathrm{SolSet}$ & \parbox[t]{\TableWidth}{A set that at the end of the algorithm contains all of the solutions (data structures $\mathrm{d}$) that satisfy the KKT conditions \eqref{eq:PowerAlocation-KKTNecessaryCond}.}\\
\hline
\end{tabular}
\caption{Description of the important variables in Algorithms~\ref{alg:Solve_KKT_Conditions} and~\ref{alg:Solve_KKT_RecursionMainPart}.}
\label{tab:DescriptionImptVariables_in_Algs}
\end{table}

In the following, in Lemma~\ref{lem:PowerAllocation_Solution_All_I_Positive} and in \S\ref{sec:HighDynamic_HighSNR_Regime}, we present two special cases for the optimization problem \eqref{eq:PowerAloocation-P2} that is insightful.

\begin{lemma}\label{lem:PowerAllocation_Solution_All_I_Positive}
Consider the set of $\{r^{(1)}_k\}_{k=0}^s$ as defined in \eqref{eq:PowerAllocation_LinearCase_root}. If we have $0= r^{(1)}_s< r^{(1)}_{s-1} < \cdots < r^{(1)}_1 < r^{(1)}_0=P_{\mathrm{max}},$
then the KKT conditions given by \eqref{eq:PowerAlocation-KKTNecessaryCond} have a unique solution. Moreover, the optimal power allocation is determined by $I_k^{**}=r^{(1)}_k$ for all $k\in[1:s-1]$.
\end{lemma}
\begin{proof}
For the proof of this lemma refer to Appendix~\ref{apn:SomeProofs}. 
\end{proof}

As an example, based on the above discussions, for the case of having 3 states ($s=2$) the solution of the power allocation can be simplified as presented in Lemma~\ref{lem:OptProbSol_SpecialCase_s=2}.

\begin{lemma}\label{lem:OptProbSol_SpecialCase_s=2}
If $s=2$ then based on the values of channel coefficients and 
probability distribution over the states, we have the following result for the optimal power allocation:
\begin{enumerate}
	\item if $\alpha_1<\beta_1$ then we have $P^{**}_1=P_{\mathrm{max}}$,
	\item if ${h}_2\beta_1 < {h}_0\alpha_1$ then $P^{**}_2=P_{\mathrm{max}}$, and finally
	\item if $\alpha_1>\beta_1$ and ${h}_2\beta_1 > {h}_0\alpha_1$ then $P^{**}_2=\min\left(r^{(1)}_1,P_{\mathrm{max}}\right)$ where $r^{(1)}_1$ is defined in \eqref{eq:PowerAllocation_LinearCase_root}.
\end{enumerate}
\end{lemma}

\subsection{High-dynamic range, high-SNR regime} \label{sec:HighDynamic_HighSNR_Regime}
In this section, we show that our proposed achievability scheme, stated in Section~\ref{sec:GrpSecretKey-GaussBrdcstChnl-LowerBound}, is optimal for the ``high-dynamic range, high-SNR regime'' in a degrees of freedom sense. We first give a formal definition of degrees of freedom in our setup as follows.
The degrees of freedom for secret key sharing over a state dependent Gaussian broadcast channel is defined as
\begin{align*}
\DoF_s  = \lim_{Q\rightarrow \infty} \frac{C^{\mathsf{gaus}}_s}{\frac{1}{2} \log Q} 
\end{align*}
where $h_i = Q^{\gamma_i}$ for $i \in [0:s]$, $\gamma_i > 0$ and $\gamma_i > \gamma_{i-1}$.

Clearly, as $Q \rightarrow \infty$, $h_i \gg h_{i-1}$ (high-dynamic range) and $h_i \gg 1$ (high-SNR). The following theorem completely characterizes $\DoF_s$, and hence proves the optimality of our proposed achievability scheme in the high-dynamic range and high-SNR regime.

\begin{theorem}\label{thm:dof_characterization}
The degrees of freedom ($\DoF_s$) for secret key sharing over a state dependent Gaussian broadcast channel is given by:
$\DoF_s = L \sum_{i=1}^{s} \left (\gamma _{i}-\gamma_{i-1}  \right) \Delta_i$.
\end{theorem}
\begin{proof}
We prove the theorem in two steps. First, we show a lower bound on $\DoF_s$ using the proposed achievability scheme in Section~\ref{sec:GrpSecretKey-GaussBrdcstChnl-LowerBound}, and then we show a matching upper bound on $\DoF_s$ using the upper bound on the secret key generation capacity as stated in Theorem~\ref{thm:SecKeyGenCapacity_UpperBound_Gaussian}.

\paragraph*{Lower bound on $\DoF_s$} If ${h}_i\gg {h}_{i-1}$ for all $i$, we have $r^{(1)}_i \stackrel{\cdot}{=} Q^{-\gamma_i}$.
Then the ordering condition stated in Lemma~\ref{lem:PowerAllocation_Solution_All_I_Positive} 
is satisfied and as a result we have $I^{**}_i= r^{(1)}_i$. Using this observation, we can derive a lower bound on $\DoF_s$ as shown below,
\begin{align}
\DoF_s  
&\stackrel{(a)} \geq \lim_{Q\rightarrow \infty}  \frac{ L\sum_{i=1}^s \Delta_i \left( \frac{1}{2} \log\left(\frac{1+ {h}_i r^{(1)}_{i-1}}{1+ {h}_i r^{(1)}_i} \cdot \frac{1+ {h}_{i-1} r^{(1)}_i}{1+ {h}_{i-1} r^{(1)}_{i-1}} \right) \right)} {\frac{1}{2} \log Q}  \nonumber\\
&= L\sum_{i=1}^s \Delta_i \left( \gamma_{i}-\gamma_{i-1} \right) 
\label{eq:dof_lower_bound}
\end{align}
where (a) follows from ${h}_i\gg {h}_{i-1}$ and Lemma~\ref{lem:PowerAllocation_Solution_All_I_Positive}.

\paragraph*{Upper bound on $\DoF_s$} An upper bound on $\DoF_s$ can be derived as shown below,
\begin{align*}
\DoF_s  
	& \stackrel{\text{(a)}}{\le} L \sum_{i>j} \delta_i\delta_j \left (\gamma_{i}-\gamma_{j} \right ) \nonumber\\
	&= L \sum_{i=1}^{s}\sum_{j=0}^{i-1} \sum_{k=j+1}^{i} \left (\gamma_{k}-\gamma_{k-1} \right )  \delta_i\delta_j  \nonumber\\
	&\stackrel{\text{(b)}}{=} L \sum_{k=1}^{s}\sum_{i=k}^{s}\sum_{j=0}^{k-1}(\gamma_{k}-\gamma_{k-1})\delta_i\delta_j \nonumber\\
	&=  L \sum_{k=1}^{s} \Delta_k (\gamma_{k-1}-\gamma_{k})
\end{align*}
where (a) follows from Theorem~\ref{thm:SecKeyGenCapacity_UpperBound_Gaussian} and (b) follows by exchanging the order of the summations. The above upper bound on $\DoF_s$ matches the lower bound in \eqref{eq:dof_lower_bound} and this completes the proof of the theorem.
\end{proof}

\begin{algorithm}[H]
\caption{The recursive core part of the algorithm that is called by Algorithm~\ref{alg:Solve_KKT_Conditions}.} \label{alg:Solve_KKT_RecursionMainPart}
\begin{algorithmic}[1]  
  \Function{Recursion}{$\mathrm{d}$,$\{h_i\}_{i=0}^s$, $\{\Delta_k\}_{k=1}^s$}
    \State $\mathrm{SolutionSet} = \varnothing$
    \State Find an index $j$ such that $\mathrm{d}(j).\mathrm{determined}=\mathrm{false}$ \label{alg:line:Pick_a_k} 
    \If{there is such $j$}
      \State $k=\mathrm{d}(j).\mathrm{l}$
      \If{$\mathrm{d}(j).\mathrm{l} = \mathrm{d}(j).\mathrm{u}$} \Comment{Linear case ($I^*_k\neq I^*_{k+1}$)} \label{alg:line:Linear_Case}
        \State Find the root $r^{(1)}$ of the numerator of $F^{(1)}_k(x)$ defined in \eqref{eq:PowerAllocation_LinearCase_root}
        \State Based on the value of $r^{(1)}$, break the interval $[\mathrm{d}(j).\mathrm{min}, \mathrm{d}(j).\mathrm{max}]$ if necessary
        \ForAll{possible subinterval of the interval $[\mathrm{d}(j).\mathrm{min}, \mathrm{d}(j).\mathrm{max}]$}
          \State Find the sign of $F^{(1)}_k(x)$ in this subinterval
          \State According to the sign of $F^{(1)}_k(x)$ (and based on Case~\ref{cond:PowerAloc_KKT_LinearCase_Conds}), update $\mathrm{d}(j).\mathrm{l}$, $\mathrm{d}(j).\mathrm{u}$, $\mathrm{d}(j).\mathrm{min}$, $\mathrm{d}(j).\mathrm{max}$, and $\mathrm{d}(j).\mathrm{determined}$, but to a new data structure $\mathrm{d}'$
          \If{$\mathrm{d}'$ is consistent up to this point} \Comment{if $\forall i$ we have $\mathrm{d}'(i).\mathrm{min} \le \mathrm{d}'(i).\mathrm{max}$} \label{alg:line:Chk_Consistency_LinCase}
            \State $\mathrm{SolutionSet} \gets \mathrm{SolutionSet}\ \cup$ \Call{Recursion}{$\mathrm{d}'$, $\{h_i\}_{i=0}^s$, $\{\Delta_k\}_{k=1}^s$} \label{alg:line:Call_Recrsion_LinCase}
          \EndIf
        \EndFor
      \Else \Comment{Quadratic case}  \label{alg:line:Quadratic_Case}
        \State Find the roots $r^{(2)}_1$ and $r^{(2)}_2$ of the numerator of $F^{(2)}_k(x)$ 
        \State Based on the values of $r^{(2)}_1$ and $r^{(2)}_2$, break the interval $[\mathrm{d}(j).\mathrm{min}, \mathrm{d}(j).\mathrm{max}]$ if necessary
        \ForAll{possible subinterval of the interval $[\mathrm{d}(j).\mathrm{min}, \mathrm{d}(j).\mathrm{max}]$}
          \State Find the sign of $F^{(2)}_k(x)$ in this subinterval
          \State According to the sign of $F^{(2)}_k(x)$, update $\mathrm{d}(j).\mathrm{l}$, $\mathrm{d}(j).\mathrm{u}$, $\mathrm{d}(j).\mathrm{min}$, $\mathrm{d}(j).\mathrm{max}$, and $\mathrm{d}(j).\mathrm{determined}$, but to a new data structure $\mathrm{d}'$
          \If{$\mathrm{d}'$ is consistent up to this point} \Comment{if $\forall i$ we have $\mathrm{d}'(i).\mathrm{min} \le \mathrm{d}'(i).\mathrm{max}$} \label{alg:line:Chk_Consistency_QuadCase}          
            \State $\mathrm{SolutionSet} \gets \mathrm{SolutionSet}\ \cup$ \Call{Recursion}{$\mathrm{d}'$, $\{h_i\}_{i=0}^s$, $\{\Delta_k\}_{k=1}^s$} \label{alg:line:Call_Recrsion_QuadCase}
          \EndIf
        \EndFor      
      \EndIf      
    \Else \Comment{If all $\mathrm{d}(j)$'s are determined}
    	  \If{the found solution is consistent} \Comment{if $\forall i$ we have $\mathrm{d}(i).\mathrm{min} \le \mathrm{d}(i+1).\mathrm{min}$}
    	    \State $\mathrm{SolutionSet} = \left\{ \mathrm{d} \right\}$ \label{alg:line:Dtermined_and_Satisfied}
    	  \EndIf
    \EndIf
    \State \Return $\mathrm{SolutionSet}$
  \EndFunction
\end{algorithmic}
\end{algorithm}

\subsection{Numerical Evaluations} \label{subsec:GrpSecGauss-numerical_evaluation}
In this section, we numerically evaluate the performance of the secret key sharing scheme proposed in Section~\ref{sec:GrpSecretKey-GaussBrdcstChnl-SolvePowerAllocProblem} for a few examples and compare it with the upper bound stated in Theorem~\ref{thm:SecKeyGenCapacity_UpperBound_Gaussian}.

\begin{example}\label{ex:Numerical_Evaluation_Ex1}
Consider a setup with $3$ states ($s=2$) where $h_0 = -5\mathrm{dB}$, $-5\mathrm{dB} < h_1<30\mathrm{dB}$ and $h_2 = 30\mathrm{dB}$.
The probability distribution across the states is assumed to be uniform. Figure~\ref{fig:three_states_UB_vs_LB} shows the achievable rate and the upper bound as a function of $h_1$ with the following choices of $P_{\mathrm{max}}$: (a) $P_{\mathrm{max}}= 0.01$ and (b) 
$P_{\mathrm{max}}= 10$. Clearly, there is a gap between the upper bound and the achievable rate. As it is mentioned before, the proposed scheme is not optimal in an absolute sense, but only in a degrees of freedom sense as proved in \S\ref{sec:HighDynamic_HighSNR_Regime}.

\begin{figure}
\centering
\includegraphics[scale=0.5]{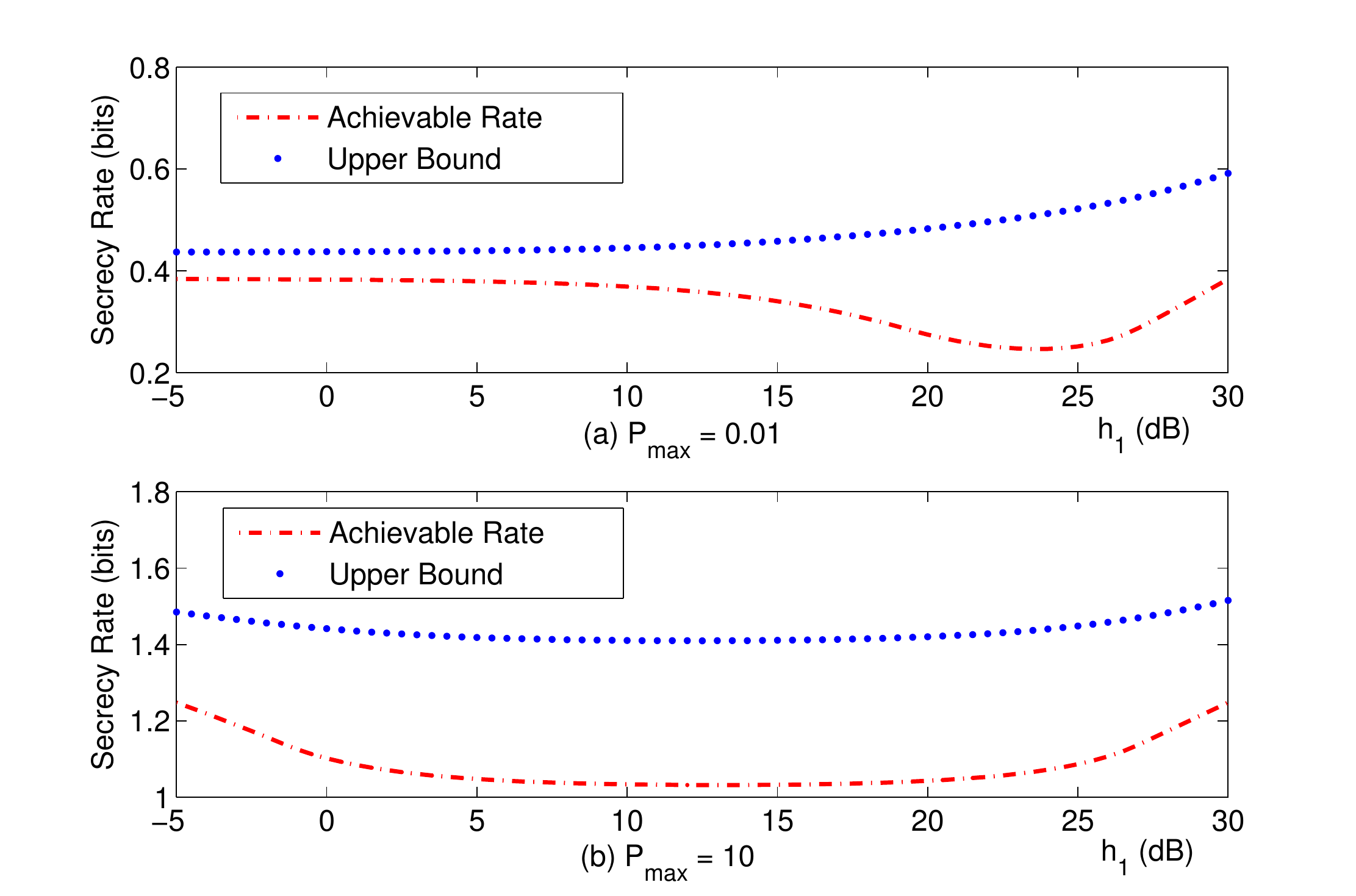}
\caption{The achievable rate and the upper bound as a function of $h_1$ with $P_{\mathrm{max}}$: (a) $P_{\mathrm{max}}= 0.01$, (b) $P_{\mathrm{max}}= 10$ (see Example~\ref{ex:Numerical_Evaluation_Ex1}).}
\label{fig:three_states_UB_vs_LB}
\vspace{-15pt}
\end{figure} 
\end{example}

\begin{example}\label{ex:Numerical_Evaluation_Ex2}
Consider a setup with $4$ states where $h_0 = -5\mathrm{dB}$, $h_3 = 30\mathrm{dB}$, $h_1 = \min[g_1, g_2]\ \mathrm{dB}$ and $h_2 = \max[g_1, g_2]\ \mathrm{dB}$ where $-5\mathrm{dB}$ $< g_1,g_2<30\mathrm{dB}$.
The probability distribution across the states is assumed to be uniform.
Figure~\ref{fig:four_states_UB_vs_LB} shows the achievable rate and the upper bound as a function of $g_1$ and $g_2$ with $P_{\mathrm{max}}=10$. Similar to Example~\ref{ex:Numerical_Evaluation_Ex1}, this illustrates the absolute gap between the upper bound and the achievable rate.
\begin{figure}
\centering
\includegraphics[scale=0.3]{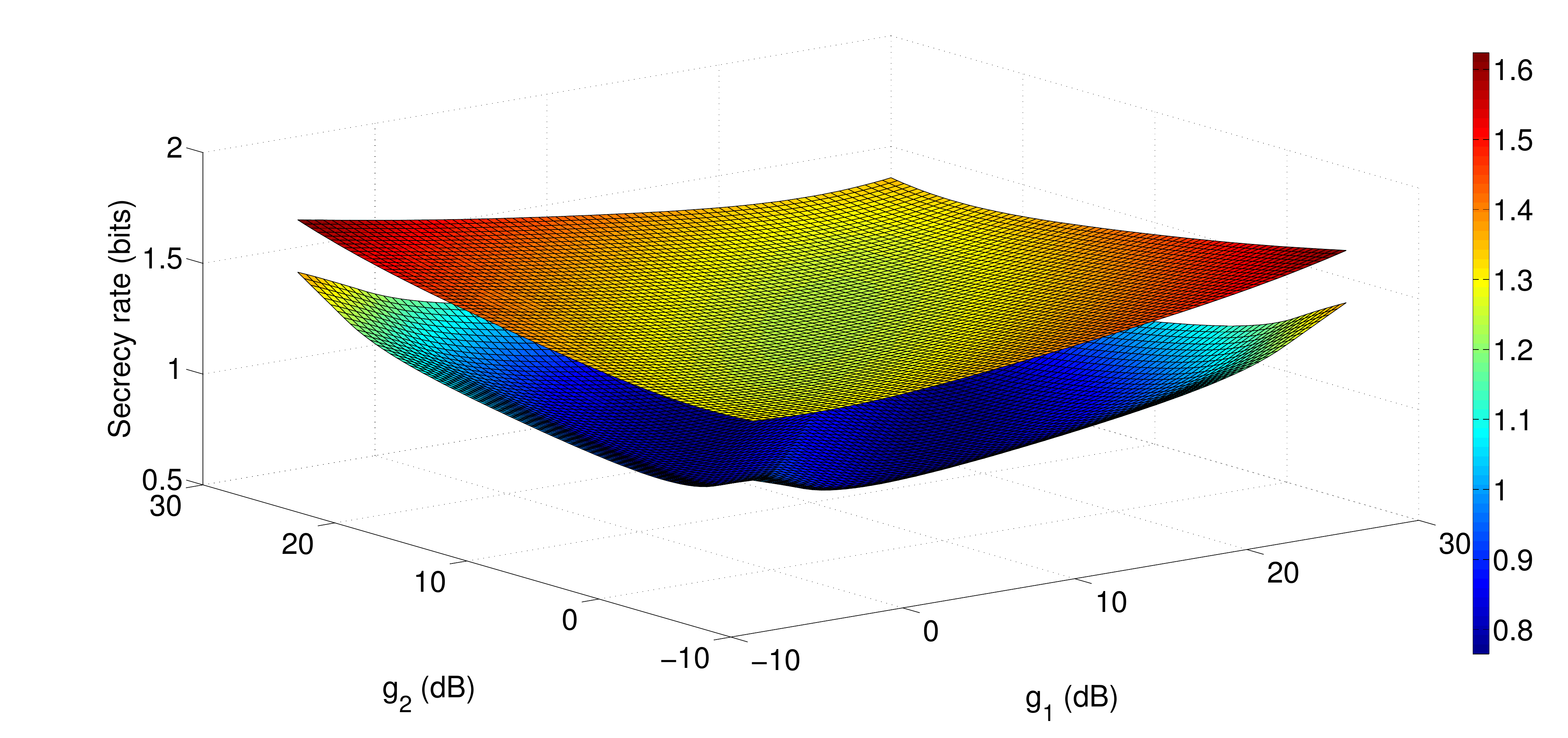}
\caption{The achievable rate (lower surface) and the upper bound (upper surface) as a function of $g_1$ and $g_2$ with $P_{\mathrm{max}}=10$ in a setup with $4$ equiprobable states (see Example~\ref{ex:Numerical_Evaluation_Ex2}).}
\label{fig:four_states_UB_vs_LB}
\vspace{-15pt}
\end{figure} 
\end{example}

\begin{example}\label{ex:Numerical_Evaluation_Ex3}
Consider a setup with $36$ (equiprobable) states, uniformly spaced in the range $-5\mathrm{dB}$ to $30\mathrm{dB}$ (\ie, $h_i= (-5+i)\mathrm{dB}$, $i\in[0:s]$). Figure~\ref{fig:power_profile_5_to30} shows the fraction of $P_{\mathrm{max}}$ allocated to each state by the proposed scheme for $P_{\mathrm{max}}\in\{0.1, 1, 10, 100\}$.
\begin{figure}
\centering
\includegraphics[scale=0.35]{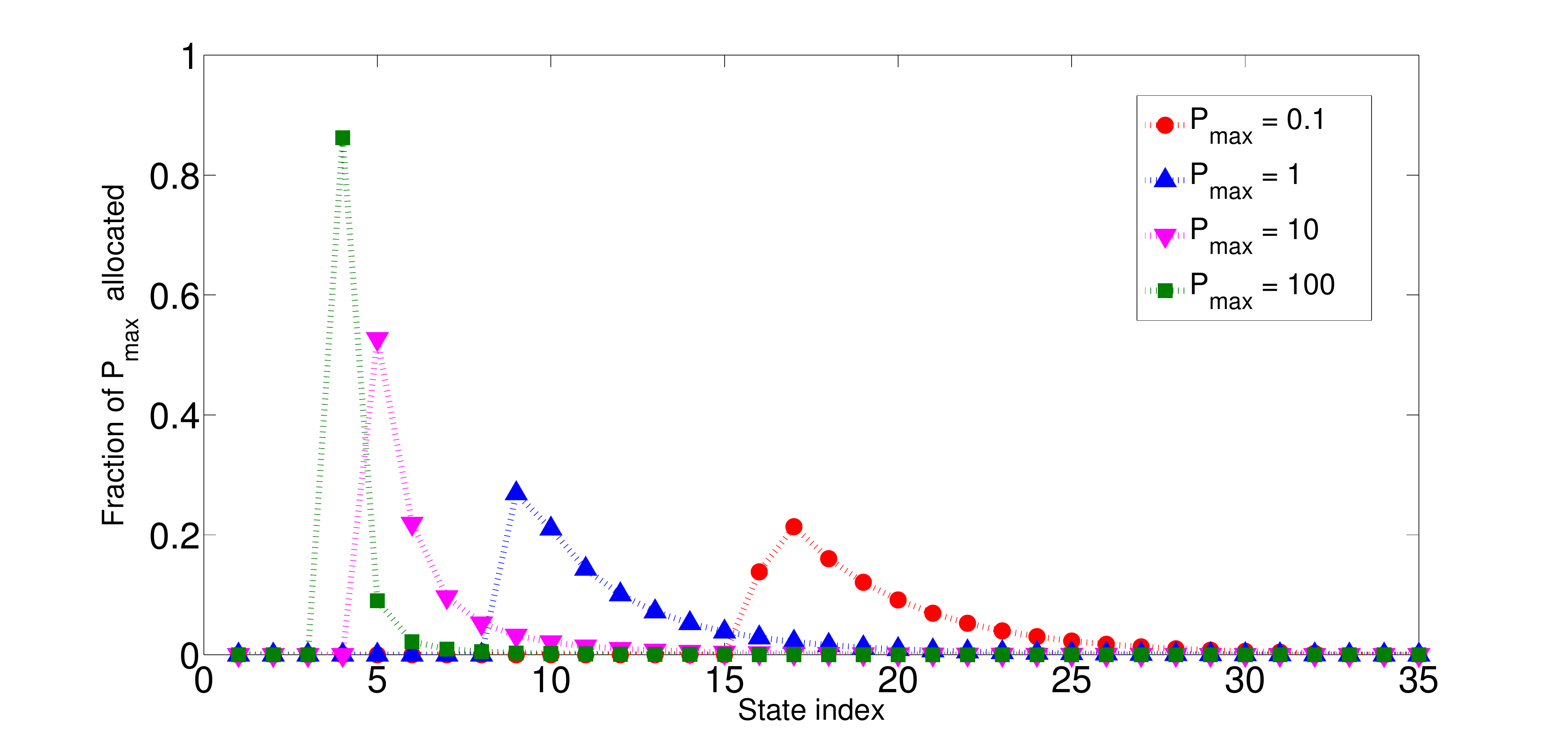}
\caption{Fraction of $P_{\mathrm{max}}$ allocated to each layer by the proposed scheme as explained in Example~\ref{ex:Numerical_Evaluation_Ex3}.
}
\label{fig:power_profile_5_to30}
\vspace{-10pt}
\end{figure}
\end{example}

The above examples illustrate different aspects of the proposed scheme: the gap with respect to the upper bound and the distribution of power across the states. In general, these aspects depend on the setup parameters. For a given setup, the numerical implementation of our proposed scheme can be used for efficient evaluations, even with a large number of states (\eg, $36$ states in Example~\ref{ex:Numerical_Evaluation_Ex3}).

\section{Discussion, Open Questions and Future Directions}
\label{sec:OpenQuestion_FutureDirections}
Here, in this section we bring forward discussion about multi-party secret key sharing problem, open questions and possible future directions. 

First, the SKG capacity problem among multiple terminals over a state-dependent Gaussian channel in the presence of a passive eavesdropper is still unsolved.
But, the optimality of the proposed multi-party secret key sharing scheme has been shown for deterministic channels (that includes the erasure channels as a special case). By having intuition from this result, the achievability scheme for the Gaussian state-dependent channel is based on the message level erasure, simulated by using the wiretap code. However, in our outer bound on the SKG capacity, we do not have such an assumption and this can be a reason explains the gap between our achievability scheme and outer bound. 

Similar ideas used in this work for the secret key sharing problem over erasure channels can also be applied for the secret communication over these channels, \eg, see \cite{IT15-Laszlo-Secrecy}. However, in our work, we go beyond and used these ideas to propose a coding scheme for multi-terminal secret key sharing over the Gaussian state-dependent broadcast channel (in the presence of public discussion). On the other hand, this is still open whether the same connection can be obtained between secret communication over erasure and state-dependent Gaussian channels or not.

In our achievability scheme, we use public channel to send feedback from all the receivers to Alice. However, it is worth mentioning that although the public channel is available and without cost, we use it to communicate only the channel state which is a limited feedback; but not to transmit all the output feedback. 
Hence, it is possible to adopt our protocol to use ACK/NAK (\eg, similar to \cite{argyraki_creating_2013}) instead of public channel. However, the resulting protocol maybe not optimal even for the deterministic channels.

Finally, we would like to emphasize that this thread of work is not pure theoretical and there have been some attempts to implement these ideas (\eg, see \cite{safaka_exchanging_2011,safaka_exchanging_2013, atsan_low_2013,argyraki_creating_2013}). As an example, \cite{argyraki_creating_2013} reports to create shared secret key in a test-bed containing 5 nodes at rate 10 kbit/sec, with their secrecy being independent of the adversary's computational capabilities.

\appendices
\section{Group Secret Key Agreement over Erasure Broadcast Channels}
\label{sec:GrpSecKeyAgr_ErasureBrdcstChnl}
In this appendix, we characterize the secret key generation capacity among multiple terminals communicating over an erasure broadcast channel. More precisely we prove the following result.

\begin{theorem}[{\cite{siavoshani_group_2010}}]\label{thm:GrpScrtKey-ErsrChnl}
The secret key generation capacity among $m$ terminals that have access to an erasure broadcast channel
is given by
\begin{equation*}
C^{\mathsf{ers}}_s = (1-\delta)\delta_\eve \left(\pktlen\log{q}\right),
\end{equation*}
where $\delta$ is the erasure probability from Alice to the rest of terminals and $\delta_\eve$ is the erasure probability from Alice to Eve. It worth to mention that the capacity $C^{\mathsf{ers}}_s$ does not depend on $m$.
\end{theorem}

The converse part of Theorem~\ref{thm:GrpScrtKey-ErsrChnl} is a direct corollary of Theorem~\ref{thm:SecKeyGenCapacity_UpperBound_Deterministic} assuming that we have only two states; \ie, complete erasure and complete reception. The achievability part is stated in \S\ref{sec:GrpScrtKey-ErsrChnl-AchvScheme}.
\subsection{Lower Bound for the Key Generation Capacity}\label{sec:GrpScrtKey-ErsrChnl-AchvScheme}
Here we describe and analyse our proposed achievability scheme for the secret key generation over erasure broadcast channels. The proposed scheme achieves secrecy rate that matches the upper bound derived from Theorem~\ref{thm:SecKeyGenCapacity_UpperBound_Deterministic}. Moreover, the complexity of this scheme is polynomial in the packet length $L$ and block length $n$. The proposed scheme consists of several phases and it proceeds as follows.

\noindent{\bf \small{Private Phase:}}
\begin{enumerate}
\item
Alice broadcasts $n$ packets, ${x}_1,\ldots,{x}_n,$ where 
${x}_i\in\Fbb_q^{\pktlen}$ and ${x}_i\sim\Uni{\Fbb_q^\pktlen}$ 
(we will call them ``$x$-packets''). Of these, $n^*$ packets are received 
by at least one honest node. This set is denoted by $\mc{N}^*$ where $n^*=|\mc{N}^*|$.
\end{enumerate}

\noindent{\bf \small{Public Discussion (Initial Phase):}}
\begin{enumerate}
\item
Each honest node sends Alice publicly a feedback message specifying which 
$x$-packets it received. Let $\set{I}_{i}$ denotes the set of packets' indices
received by the $i$th terminal where $i\in[1:m-1]$. 
\item 
Alice constructs $h = \delta_\eve \cdot n^*$ linear combinations of the $x$-packets, ${y}_1, \ldots, {y}_h$ (we will call them ``$y$-packets''), as follows:

\textsf{(i)} She divides the set $\set{N}^*$ of $x$-packets that were received by at least one honest 
node into non-overlapping subsets, such that each subset consists of all the 
packets that were commonly received by a different subset of honest nodes.
To be more precise, let $\set{S}$ be an arbitrary non-empty subset of $[1:m-1]$ and let us define the set
\begin{equation*}
\set{N}_{\set{S}} \triangleq \big\{ i\in[1:n]\ |\ i\in \set{I}_{j}:\forall j\in \set{S},\text{ and } i\notin \set{I}_{j}:\forall j\in \overline{\set{S}} \big\}.
\end{equation*}
Then we have
\begin{equation*}
\set{N}^* = \bigcup_{\varnothing \neq \set{S}\subseteq [1:m-1]} \set{N}_{\set{S}}.
\end{equation*}
Figure~\ref{fig:BobCalvin_RcvdSets} shows all the sets $\set{N}_{\set{S}}$, $\set\subseteq [1:2]$ for the case where $m=2$.

\begin{figure}
\centering
\includegraphics[scale=0.9]{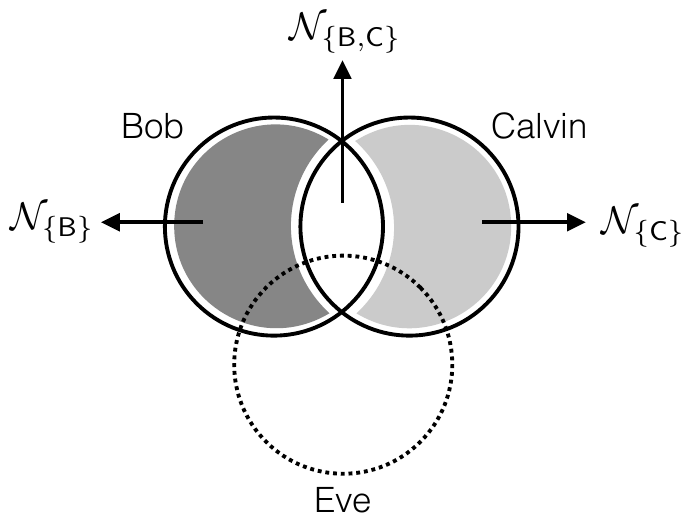}
\caption{Here, we use Bob and Calvin to denote for $\ter_1$ and $\ter_2$, respectively. The sets $\set{N}_{\{\bob\}}, \set{N}_{\{\calvin\}}$, and $\set{N}_{\{\bob,\calvin\}}$ are shown for the case where $m=2$. For this case we have $\set{N}^* = \set{N}_{\{\bob\}} \cup \set{N}_{\{\calvin\}} \cup \set{N}_{\{\bob,\calvin\}}$.}
\label{fig:BobCalvin_RcvdSets}
\end{figure}

\textsf{(ii)} From each such subset of packets $\set{N}_{\set{S}}$, she creates 
$\delta_\eve \cdot n_{\set{S}}$ linear combinations using the construction described in Lemma~\ref{lemma_linear_1} (provided in \Apn\ref{sec:GrpScrtKey-Apndx-SomeLemmas}), 
where $n_{\set{S}} \triangleq|\set{N}_{\set{S}}|$.
Then she publicly reveals the coefficients she used to create all the $y$-packets.
\item 
Each node $i\in[1:m-1]$ reconstructs as many (say $h_i$) of the $y$-packets as it 
can (based on the $x$-packets it received in step \#1). For $h_i$ we can write
\begin{equation*}
h_i =  \sum_{\begin{subarray}{c}\varnothing \neq \set{S}\subseteq [1:m-1]:\ i\in \set{S} \end{subarray}} \delta_\eve\cdot n_{\set{S}}.
\end{equation*}
Note that as $n$ grows we have $h_i\rightarrow\Expc{h_i}$ which is equal to
$\Expc{h_i} = (1-\delta)\delta_\eve n$.
Up to this point of the protocol, and assuming that $n$ is large enough, for each subset $\set{S}\subseteq [1:m-1]$ the terminals in the set $\set{S}\cup \{0\}$ have created a common secret key (which is secure with high probability) of length $\delta_\eve n_{\set{S}}$ among themselves.
\end{enumerate}

\noindent{\bf \small{Public Discussion (Reconciliation Phase):}}
\begin{enumerate}
\item 
Alice creates $h-\min_i h_i$ linear combinations of the $y$-packets 
(we will call them ``$z$-packets''), using the construction provided in 
Lemma~\ref{lemma_linear_3} (provided in \Apn\ref{sec:GrpScrtKey-Apndx-SomeLemmas}).
She publicly reveals both the contents and the coefficients of the 
$z$-packets, such that each node $i$ receives at least $h-h_i$ of them.
\item
Each node $i$ combines the ${h-h_i}$ $z$-packets it received with the 
$h_i$ $y$-packets it recreated in phase 1, and reconstructs all the $y$-packets.
\item 
Alice creates $l = \min_i h_i$ linear combinations of the $y$-packets, $k_1, \ldots, k_l$
(we will call them ``$k$-packets''), using the construction stated in 
Lemma~\ref{lemma_linear_2} (provided in the \Apn\ref{sec:GrpScrtKey-Apndx-SomeLemmas}).
She publicly reveals the coefficients she used to create all the $k$-packets. 
\item
Each node $i$ reconstructs all the $k$-packets. The common secret key is the 
concatenation of all the $k$-packets, $K = \{k_1, \ldots, k_l \}$.
\end{enumerate}

Now, based on Definition~\ref{def:GrpSecKey-SecKeyGenProtocol}, we summarize the above achievability scheme as follows. At $t=0$ Alice generates the random variable $Q_0=\{x_1,\ldots,x_n \}$ where ${x}_i\sim\Uni{\Fbb_q^\pktlen}$. We have also $Q_{1:{m-1}}=\varnothing$.
For each time $t$, $1\le t<n$, she broadcasts ${X}_0[t]={x}_t$ and there is no public discussions afterwards; namely we have ${D}[t]=\varnothing$. After the $n$th transmission by Alice there is a public discussion in many rounds; summarized as follows. We have ${D}[n]=(P_1,P_2,P_3,P_4)$ where $P_1$ denotes the set of indices $\set{I}_{i}$ that have been sent back by the honest terminals, $P_2$ denotes the coefficients of the $y$-packets, $P_3$ denotes the $z$-packets and their coefficients, and finally $P_4$ represents the coefficients of $k$-packets.

The performance of the proposed protocol in terms of the secret key generation rate, is stated in the following theorem.

\begin{theorem}\label{thm:AchvAecrecyRate-ErasureChnl}
The achievable secret key generation rate of the above scheme is
\begin{equation*}
{R}^{\mathsf{ers}}_s = (1-\delta)\delta_\eve \left(\pktlen\log{q}\right).
\end{equation*}
It is noteworthy to mention that the achievable rate does not depend on the number of terminals $m$.
\end{theorem}

\begin{proof}
The way that the achievability scheme is proposed, constructively satisfies  
Condition~\eqref{eq:AchvRate_Cond1}, \ie, we have
$\Prob{K_i\neq K_j} = 0$, $\forall i,j\in [0:m-1],\ i\neq j$.

To prove Condition~\eqref{eq:AchvRate_Cond3}, we proceed as follows. First, let us define $\overline{l}\triangleq l/n$. Then we can write
\begin{align*}
H(K) &= H(K,\overline{l}) = H(K|\overline{l}) + H(\overline{l}) \nonumber\\
&\ge H(K|\overline{l}) \nonumber\\
&= H(K|\alpha<\overline{l}<\beta)\cdot \Prob{\alpha<\overline{l}<\beta} + H(K|\overline{l}\ge\beta)\cdot \Prob{\overline{l}\ge\beta}
+ H(K|\overline{l}\le\alpha)\cdot \Prob{\overline{l}\le\alpha} \nonumber\\
&\ge H(K|\alpha<\overline{l}<\beta)\cdot \Prob{\alpha<\overline{l}<\beta} \nonumber\\
&\ge n\alpha \left(\pktlen\log{q}\right) \Big[1-\Prob{\overline{l}\le\alpha}-\Prob{\overline{l}\ge\beta} \Big],
\end{align*}
where $\mu=(1-\delta)\delta_\eve$ and also $\alpha=\mu-\gamma$ and $\beta=\mu+\gamma$ for some small $\gamma$ such that $0<\gamma\le \mu$.
Now by applying the concentration result of Lemma~\ref{lem:ConcentrationLemma} 
(see \Apn\ref{sec:GrpScrtKey-Apndx-SomeLemmas}), we have
\begin{equation*}
\Prob{\overline{l}\le\alpha} \le m\exp\left(-\frac{\gamma^2}{2\mu} n\right)\triangleq a,
\end{equation*}
and
\begin{equation*}
\Prob{\overline{l}\ge\beta} \le \exp\left(-\frac{m\gamma^2}{3\mu} n\right) \triangleq b.
\end{equation*}
Hence, observe that by choosing 
\begin{equation*}
{R}^{\mathsf{ers}}_s = \mu (\pktlen\log{q})
\end{equation*}
and
\begin{gather*}
\epsilon = \mu\left(\pktlen\log{q}\right) \left[ a+b \right] + \gamma \left(\pktlen\log{q}\right) \left[1- a-b \right]
\end{gather*}
we have Condition~\eqref{eq:AchvRate_Cond3} satisfied, \ie,
$\frac{1}{n}H(K) > {R}^{\mathsf{ers}}_s-\epsilon$.
Finally, we get the desired result by making $\gamma$ arbitrarily small because we have $\epsilon\rightarrow 0$ if $\gamma\rightarrow 0$.

To prove Condition~\eqref{eq:AchvRate_Cond2}, we need to show that
\begin{equation*}
I(K;{X}_\eve^n,P_1,P_2,P_3,P_4)<\epsilon.
\end{equation*}
By using a similar technique used above to bound $H(K)$, (more precisely by using Lemma~\ref{lemma_linear_1} and some concentration results for $n_{\set{S}}$), it can be shown that
\begin{equation}\label{eq:proof_achv_secrecy_rate_1}
I({Y};{X}_\eve^n,P_1,P_2) < \epsilon
\end{equation}
where $Y$ is a random variable representing the $y$-packets.
Using Lemma~\ref{lemma_linear_2}, by construction, we have also
\begin{equation}\label{eq:proof_achv_secrecy_rate_2}
I(K ; P_3,P_4)=0.   
\end{equation}
Now we know that the coefficients of the $z$-packets and $k$-packets form
a basis (see Lemma~\ref{lemma_linear_2}) so the random variable ${Y}$ and the 
random variable $(K,P_3,P_4)$ are equivalent (having one we have the other).
Then we can write \eqref{eq:proof_achv_secrecy_rate_1} as follows
\begin{align*}
I({Y} ; {X}_\eve^n,P_1,P_2) &= I(K,P_3,P_4 ; {X}_\eve^n,P_1,P_2) \nonumber\\
&= I(P_3,P_4 ; {X}_\eve^n,P_1,P_2) + I(K ; {X}_\eve^n,P_1,P_2 | P_3,P_4) < \epsilon,
\end{align*}
so 
\begin{equation}\label{eq:proof_achv_secrecy_rate_3}
I(K ; {X}_\eve^n,P_1,P_2 | P_3,P_4) < \epsilon.   
\end{equation}
Now we can expand
\begin{align*}
I(K ; {X}_\eve^n,P_1,P_2,P_3,P_4) &= I(K ; P_3,P_4) + I(K ; {X}_\eve^n,P_1,P_2 | P_3,P_4),
\end{align*}
where the first term is zero by \eqref{eq:proof_achv_secrecy_rate_2} 
and second term is very small because of \eqref{eq:proof_achv_secrecy_rate_3}. This concludes the theorem. 
\end{proof}

Theorem~\ref{thm:AchvAecrecyRate-ErasureChnl} in addition to the upper bound derived from Theorem~\ref{thm:SecKeyGenCapacity_UpperBound_Deterministic}, complete the proof of Theorem~\ref{thm:GrpScrtKey-ErsrChnl}.


\subsection{Some Lemmas}\label{sec:GrpScrtKey-Apndx-SomeLemmas}
\begin{lemma}\label{lemma_linear_1}
Consider a set of $n$ packets ${x}_1,\ldots,{x}_n$, ${x}_i\in\Fbb_q^\pktlen$, 
where ${x}_i\sim\Uni{\mathbb{F}_q^\pktlen}$ and  all the packets
${x}_i$ are independent from each other. Assume that Eve has overheard
$n_\eve$ of these packets. Call the packets Eve has ${w}_1,\ldots,{w}_{n_\eve}$. Then 
it is possible to create $n' = n-n_\eve$ linear combinations of the ${x}_1,\ldots,{x}_n$ 
packets over the finite field $\mathbb{F}_q$, 
say ${y}_1,\ldots,{y}_{n'}$, in polynomial time, so that these $n'$ new $y$-packets are secure from Eve, \ie,
\begin{equation}
I({y}_1,\ldots,{y}_{n'} ; {w}_1,\ldots,{w}_{n_\eve}) = 0.
\end{equation}
The same result holds with high probability (of order $1-O(q^{-1})$)
if the linear combinations are selected uniformly at random  over  $\Fbb_q$.
\end{lemma}
\begin{proof}
Let $\Bs{X}$ be an $n\times L$ matrix that has as rows the packets ${x}_1,\ldots,{x}_n$. 
Similarly, construct matrices $\Bs{Y}$ and $\Bs{W}$ that have as rows the packets 
${y}_1,\ldots,{y}_{n'}$ and ${w}_1,\ldots,{w}_{n_\eve}$.

Note that because the packets ${w}_1,\ldots,{w}_{n_\eve}$ are by definition 
a subset of the packets ${x}_1,\ldots,{x}_n$, we can write  $\Bs{W}=\Bs{A}_\eve \Bs{X}$,
with $\Bs{A}_\eve\in\Fbb_q^{n_\eve\times n}$ that has zeros and ones as elements.
We can also write the ${y}$-packets as linear combinations of the ${x}$-packets 
over the finite field $\Fbb_q$. We will then have that $\Bs{Y}=\Bs{A} \Bs{X}$, where 
$\Bs{A}\in\Fbb_q^{(n-n_\eve)\times n}$ is the matrix we are interested in designing. 
Thus we can write
\[
\begin{bmatrix} \Bs{Y}\\ \Bs{W}  \end{bmatrix} =
\begin{bmatrix} \Bs{A}\\ \Bs{A}_\eve \end{bmatrix} \Bs{X}.
\]
We now proceed by expanding $H(\Bs{Y}|\Bs{W})$. We have
\begin{align*}
H(\Bs{Y}| \Bs{W}) &= H(\Bs{Y},\Bs{W})-H(\Bs{W}) \nonumber\\
 &= \big[ \rank\left( \Bs{B} \right) - \rank(\Bs{A}_\eve) \big] \pktlen\log{q} \nonumber\\
 &= \big[ \rank\left( \Bs{B} \right)  - n_\eve\big] \pktlen\log{q},
\end{align*}
where $\Bs{B}= \left[ \begin{smallmatrix} \Bs{A}\\ \Bs{A}_\eve \end{smallmatrix} \right]$ and $\pktlen$ is the length of each packet ${x}_i$.
Now the only way that we have $H(\Bs{Y}|\Bs{W})=H(\Bs{Y})$ is that $\Bs{B}$
becomes a full rank matrix.

Using coding theory, we will construct such a matrix $\Bs{B}$, {\em without}
knowing $\Bs{A}_\eve$.  All we know is that in each row of $\Bs{A}_\eve$ there is only
one ``$1$'' and the remaining elements are zero; so all of the vectors in
the row span of $\Bs{A}_\eve$ have Hamming weight
less than or equal to $n_\eve$. Now, if we choose $\Bs{A}$ to be a generator
matrix of an maximum distance separable (MDS) linear code
with parameters $[n,n-n_\eve,n_\eve+1]_q$ then each codeword has Hamming weight
larger than or equal to $n_\eve+1$  (\eg, see \cite{macwilliams_theory_1978}). So the row span of $\Bs{A}$
and $\Bs{A}_\eve$ are disjoint (except for the zero vector) and the
matrix $\Bs{B}$ becomes full-rank for all of matrices $\Bs{A}_\eve$ 
that have the aforementioned structure.
For example, we may select to use a generator 
matrix of a Reed-Solomon code (\eg, see \cite{macwilliams_theory_1978}), which is an MDS code, over a field of size $q=n+1$.

To prove the second assertion of the lemma, we note that creating vectors ${y}_i$ uniformly at random is equivalent to selecting the elements of matrix $\Bs{A}$ independently and uniformly at random from the field $\mathbb{F}_q$.
In this case we can write
\begin{align}
\Prob{\text{$\Bs{B}$ is full-rank}} &= \frac{(q^n-q^{n_\eve})\cdots(q^n-q^{n-1})}{q^{n(n-n_\eve)}} \nonumber\\
&= \left(1-q^{-(n-n_\eve)}\right)\cdots(1-q^{-1}) \nonumber\\
&= 1-O(q^{-1}),
\end{align}
which goes to $1$ as $q$ increases.
\end{proof}

\begin{lemma}\label{lemma_linear_3}
Consider packets ${y}_1,\ldots,{y}_h$, $y_i\in\mathbb{F}_q^L$, and assume that each one of $m-1$ receivers (apart from $\ter_0$) has observed a different subset of these packets, \ie, the $i$th receiver observes a subset of size $h_i$ packets. Then, we can find $h-l$ (where $l=\min_i h_i$) linear combinations of the ${y}$-packets, say ${z}_1,\ldots,{z}_{h-l}$ such that, each receiver can use its observations (of $y$-packets) and the ${z}$-packets to recover all of the ${y}$-packets.
\end{lemma}

\begin{proof}
This is a standard problem formulation in the NC literature, and any of the standard polynomial-time approaches for network code design can be used (\eg, see \cite{fragouli_network_2007}).
\end{proof}

\begin{lemma}\label{lemma_linear_2}
Consider a set of $h$ packets ${y}_1,\ldots,{y}_h$ where ${y}_i\sim\Uni{\mathbb{F}_q^\pktlen}$ and assume that an eavesdropper Eve has overheard $h-l$ linear combinations of these packets. Call the packets Eve has ${z}_1,\ldots,{z}_{h-l}$. Then 
it is possible to create $l$ linear combinations of the  ${y}_1,\ldots,{y}_h$ packets, say ${k}_1,\ldots,{k}_l$, in polynomial time, so that these are secure from Eve, \ie,
\[
I({k}_1,\ldots,{k}_l ; {z}_1,\ldots,{z}_{h-l}) = 0.
\]
The same result holds with high probability (probability of order $1-O(q^{-1})$) if the $l$ packets ${k}_i$ are created uniformly at random over  $\mathbb{F}_q$.
\end{lemma}

\begin{proof}
Similar to the proof of Lemma~\ref{lemma_linear_1}, let $\Bs{Y}$, $\Bs{Z}$ and $\Bs{K}$ be  
matrices  that have as rows the packets ${y}_1,\ldots,{y}_h$, ${z}_1,\ldots,{z}_{h-l}$ and ${k}_1,\ldots,{k}_l$. We can then write
\begin{equation*}
\begin{bmatrix} \Bs{K}\\ \Bs{Z}  \end{bmatrix} =
\begin{bmatrix} \Bs{A}_K\\ \Bs{A}_Z \end{bmatrix} \Bs{Y},
\end{equation*}
where $\Bs{A}_Z$ is a given known matrix, since we know the transmitted linear combinations, and we seek a matrix $\Bs{A}_K$ such that, the matrix $\left[\begin{smallmatrix} \Bs{A}_K\\ \Bs{A}_Z \end{smallmatrix}\right]$ is full rank. Equivalently, we seek vectors ${k}_1,\ldots,{k}_l$ that together with  ${z}_1,\ldots,{z}_{h-l}$ form a basis; we can do
this using any of standard methods, such as Gram-Schmidt orthogonalization.
\end{proof}

\begin{lemma}\label{lem:ConcentrationLemma}
The value of the parameter $l$ in Theorem~\ref{thm:AchvAecrecyRate-ErasureChnl}, which is defined to be $l=\min_i h_i$ converges exponentially fast in $n$ to its expected value.
\end{lemma}

\begin{proof}
Let us consider the {\em random} variables $h$, $h_i$, and $l$ defined in \S\ref{sec:GrpScrtKey-ErsrChnl-AchvScheme}. 
For convenience, we work with the normalized random variables $\overline{h}\triangleq h/n$, $\overline{h}_{i}\triangleq h_{i}/n$, and $\overline{l}\triangleq l/n$.  Let us also define the random variables $\eta_j^{(i)}$, $i\in [1:m-1]$ and $j\in [1:n]$ as follows
\begin{equation*}
\eta_j^{(i)} = \left\{\begin{array}{ll} 1 & \text{if the $j$th $x$-packet is received}\\ & \text{by $\ter_i$ but not by Eve},\\ 
0 & \text{otherwise}. \end{array} \right.
\end{equation*}
Then we can write $\overline{h}_i=\frac{1}{n}\sum_{j=1}^n \eta_j^{(i)}$ and we have $\mu=\mu_i\triangleq\Expc{\overline{h}_i}=(1-\delta)\delta_E$. As defined before, we have also $\overline{l}=\min_i \overline{h}_i$.

To bound $\overline{l}$, observe that for some small $\gamma$, $0<\gamma\le \mu$, we can write
\begin{align*}
\Prob{\overline{l}\ge \mu+\gamma} &= \Prob{\overline{h}_i\ge \mu+\gamma : \forall i} \nonumber\\
&= \Prob{\overline{h}_1\ge \mu+\gamma}^{m} \nonumber\\
&\le \exp\left(-\frac{m\gamma^2}{3\mu} n\right),
\end{align*}
where in the last inequality we use Chernoff bound \cite[Chapter~4]{mitzenmacher_probability_2005}.
On the other hand we can also write for $0<\gamma\le \mu$
\begin{align*}
\Prob{\overline{l}\le \mu-\gamma} &\le m\Prob{\overline{h}_1\le \mu-\gamma} \nonumber\\
& \le m \exp\left(-\frac{\gamma^2}{2\mu} n\right),
\end{align*}
so we are done.
\end{proof}

\section{Some Proofs}\label{apn:SomeProofs}

\begin{proof}[Proof of Theorem~\ref{thm:SecrecyUpBound-CsNa08-forIndpChnl}]
First, notice that we can write
\begin{align*}
H(X_{[0:m-1]}|X_\eve) &\stackrel{\text{(a)}}{=} \sum_{j=0}^{m-1} H(X_j|X_\eve,X_{[0:j-1]}) \nonumber\\
&\stackrel{\text{(b)}}{=} H(X_0|X_\eve) + \sum_{j=1}^{m-1} H(X_j|X_0),
\end{align*}
where (a) follows from the chain rule and (b) follows from the
independence of the channels.
Similarly, for every $B\subsetneq [0:m-1]$ we can expand $H(X_B|X_{B^c},X_\eve)$ 
as follows
\begin{align*}
H(X_B|X_{B^c},X_\eve) = H(X_0|X_{B^c},X_\eve) + \sum_{j\in B} H(X_j|X_0).
\end{align*}
Now, from Theorem~\ref{thm:SecrecyUpBound-CsNa08}, we know that
for every $\lambda\in\Lambda([0:m-1])$ there exists a distribution
$P_{X_0}$ such that $C_s$ is upper bounded by
{\allowdisplaybreaks[4]
\begin{align*}
C_s &\le H(X_{[0:m-1]}|X_\eve) -\sum_{B\subsetneq [0:m-1]} \lambda_B H(X_B|X_{B^c},X_\eve) \nonumber\\
  &= H(X_0|X_\eve) + \sum_{j=1}^{m-1} H(X_j|X_0)  - \sum_{B\subsetneq [0:m-1]} \lambda_B \left[ H(X_0|X_{B^c},X_\eve) + \sum_{j\in B} H(X_j|X_0) \right] \nonumber\\
  &\stackrel{\text{(a)}}{=} H(X_0|X_\eve) - \sum_{B\subsetneq [0:m-1],\ 0\in B} \lambda_B H(X_0|X_{B^c},X_\eve)  + \sum_{j=1}^{m-1} H(X_j|X_0) - \sum_{j=1}^{m-1} H(X_j|X_0) \sum_{B\subsetneq [0:m-1],\ j\in B} \lambda_B \nonumber\\
  &\stackrel{\text{(b)}}{=} H(X_0|X_\eve) - \sum_{B\subsetneq [0:m-1],\ 0\in B} \lambda_B H(X_0|X_{B^c},X_\eve) \nonumber\\
  &\stackrel{\text{(c)}}{=} \sum_{B\subsetneq [0:m-1],\ 0\in B} \lambda_B I(X_0;X_{B^c}|X_\eve),
\end{align*}}\hspace{-4pt}
where in (a) we have changed the order of summation over $j$ and $B$,
and (b) and (c) follows from \eqref{eq:LambdaPartitionDef}.
So up to here we have the following upper bound for the secret key generation rate
\begin{equation*}
C_s \le \max_{P_{X_0}} \min_{\lambda\in\Lambda([0:m-1])} \sum_{B\subsetneq [0:m-1],\ 0\in B} \lambda_B I(X_0;X_{B^c}|X_\eve).
\end{equation*}
Now, let us define $i=\arg\min_{j\in[1:m]} I(X_0;X_j|X_\eve)$. Then, notice 
that $\lambda_B=\lambda_{B^c}=1$ where $B^c=\{i\}$ is a valid choice
according to the condition of Theorem~\ref{thm:SecrecyUpBound-CsNa08},
\ie, they satisfy \eqref{eq:LambdaPartitionDef}.
Hence, for this choice of $\lambda$, the upper bound for $C_s$ is simplified to
\begin{align*}
C_s &\le \max_{P_{X_0}} \min_{j\in[1:m-1]} I(X_0;X_j|X_\eve) \nonumber\\
&\le \min_{j\in[1:m-1]} \max_{P_{X_0}} I(X_0;X_j|X_\eve).
\end{align*} 
\end{proof}

\begin{proof}[Proof of Lemma~\ref{lem:PowerAllocation_Solution_All_I_Positive}]
From Case~\ref{cond:PowerAloc_KKT_LinearCase_Conds}, we know that $r^{(1)}_k > 0$ only if $\alpha_k>\beta_k$ and ${h}_{k+1}\beta_k > {h}_{k-1}\alpha_k$. So as mentioned before $\forall k\in[1:s-1]$ we have $F^{(1)}_k(x)>0$ for $x>r^{(1)}_k$ and $F^{(1)}_k(x)<0$ for $x < r^{(1)}_k$.

Now, it can be easily checked that the solution $I^*_k= r^{(1)}_k$ satisfies 
the set of conditions stated in \eqref{eq:PowerAlocation-KKTNecessaryCond} with
$\lambda^*_k=0$. On the other hand, because $F^{(1)}_k(x)>0$ for $x> r^{(1)}_k$ and $F^{(1)}_k(x)<0$ for $x< r^{(1)}_k$, we show that every deviation of $I^*_k$ from the $r^{(1)}_k$ results in a violation of KKT conditions \eqref{eq:PowerAlocation-KKTNecessaryCond}. To this end, we proceed as follows.

Let us fix $k$. If $I^*_k > r^{(1)}_k$ then $F^{(1)}_k(I^*_k)>0$ and because we should have $F^{(1)}_k(I^*_k)+(\lambda^*_k-\lambda^*_{k+1})=0$, we can conclude that $\lambda^*_{k+1}>0$. So by the complementary slackness condition given in \eqref{eq:PowerAlocation-KKTNecessaryCond}, we should have
$I^*_k=I^*_{k+1}$. Now, two cases may happen. First, if $I^*_{k+1}< r^{(1)}_k$ 
that results in a contradiction because we have already assumed that $I^*_k> r^{(1)}_k$ and we have also $I^*_k=I^*_{k+1}$. Secondly, if $I^*_{k+1} > r^{(1)}_k > r^{(1)}_{k+1}$, then similar to the above argument we can show that $I^*_{k+1}=I^*_{k+2}$. Then we either encounter a contradiction in this step or have to continue. Finally, if we did not have any contradiction in these steps we would have $I^*_k=I^*_{k+1}=\cdots=I^*_s=0$. Now, this is a contradiction because we had assumed $I^*_k> r^{(1)}_k > r^{(1)}_s=0$. 

Similarly, it is possible to argue that for the case $I^*_k< r^{(1)}_k$ we will also encounter a contradiction.
So the \emph{unique} solution to the set of conditions 
\eqref{eq:PowerAlocation-KKTNecessaryCond} is given by $I^*_k=r^{(1)}_k$ and $\lambda^*_k=0$
and we are done.
\end{proof}


\bibliographystyle{IEEEtran}
\bibliography{MyLibrary}

\end{document}